\documentclass[a4paper,11pt]{article}
 
\usepackage{microtype}
\usepackage{mathtools,amsmath} \usepackage{amssymb,amsthm}
\usepackage[T1]{fontenc}
\usepackage{graphicx} \usepackage{xspace} \usepackage{xypic}
\usepackage[numbers]{natbib} \usepackage{hyperref,color}\usepackage{pdfpages}
\usepackage{fullpage}

\graphicspath{{figures/}}
\title{Voronoi Diagrams for a Moderate-Sized Point-Set in a Simple
  Polygon\thanks{This research was supported by the MSIT(Ministry of Science and ICT), Korea, under the SW Starlab support program(IITP-2017-0-00905) supervised by the IITP(Institute for Information \& communications Technology Promotion)}}

\author{Eunjin Oh\thanks{Pohang University of Science and Technology,
Korea. Email: {\tt{\{jin9082, heekap\}@postech.ac.kr}}} 
\and Hee-Kap Ahn\footnotemark[2]}
              
\newtheorem{theorem}{Theorem} 
\newtheorem{lemma}[theorem]{Lemma}

\newtheorem{corollary}[theorem]{Corollary}

\def\polylog{\operatorname{polylog}}
\newcommand{\ch}{\ensuremath{\textsc{ch}}}
\newcommand{\bd}{\ensuremath{\partial}}
\newcommand{\vd}{\ensuremath{\textsf{VD}}}
\newcommand{\kvd}{\ensuremath{k\textnormal{-}\textsf{VD}}}
\newcommand{\fvd}{\ensuremath{\textsf{FVD}}}
\begin{document}
\date{}
\maketitle

\begin{abstract}
  Given a set of sites in a simple polygon, a geodesic Voronoi
  diagram of the sites partitions the polygon into regions based on distances to
  sites under the geodesic metric.
  We present algorithms for computing the geodesic nearest-point,
  higher-order and farthest-point Voronoi diagrams of $m$ point sites in a
  simple $n$-gon, which improve the best known ones for $m \leq
  n/\polylog n$.  Moreover, the algorithms for the geodesic
  nearest-point and farthest-point Voronoi diagrams
  are optimal for $m \leq n/\polylog n$.  This partially
  answers a question posed by Mitchell in the Handbook of
  Computational Geometry.
\end{abstract}
\section{Introduction}
The \emph{geodesic distance} between any two points $x$ and $y$ contained in
a simple polygon is the length of the shortest path  in the polygon
connecting $x$ and $y$.
A geodesic Voronoi diagram of a set $S$ of $m$ sites contained
in a simple polygon $P$ partitions $P$ into regions based on
distances to sites of $S$ under the geodesic metric.  The \emph{geodesic
nearest-point Voronoi diagram} of $S$ 
partitions
$P$ into cells, exactly one cell per site, such that every point
in a cell has the same nearest site of $S$ under the geodesic metric.
The higher-order Voronoi diagram, also known as the order-$k$
Voronoi diagram, is a generalization of the nearest-point Voronoi diagram.
For an integer $k$ with $1\leq k\leq m-1$,
the \emph{geodesic order-$k$ Voronoi diagram} of $S$ partitions $P$
into cells, at most one cell per $k$-tuple of sites, such that every
point in a cell has the same $k$ nearest sites under the geodesic
metric. Thus, the geodesic order-$1$ Voronoi diagram is the geodesic
nearest-point Voronoi diagram.
The geodesic order-$(m-1)$ Voronoi diagram of $m$ sites
is also called the geodesic farthest-point Voronoi diagram.
The \emph{geodesic farthest-point Voronoi diagram} of $S$ partitions
$P$ into cells, at most one cell per site, such that every
point in a cell has the same farthest site under the geodesic metric.


In this paper, we study the problem of computing the geodesic nearest-point,
higher-order and farthest-point Voronoi diagrams of a set $S$ of $m$ point sites 
contained in a simple $n$-gon $P$.
Each edge of a geodesic Voronoi diagram is either a hyperbolic arc
or a line segment consisting of points equidistant from two sites under the
geodesic metric~\cite{Aronov-VD-1989,FVD,LL-KVD-2013}.
The boundary between any two neighboring cells of a geodesic Voronoi
diagram is a chain of $O(n)$ edges.
Each end vertex of the boundary is of degree 1 or 3 under
the assumption that no point in the plane is equidistant from four distinct sites
while every other vertex is of degree 2.
There are $\Theta(k(m-k))$ degree-3 vertices 
in the geodesic order-$k$
Voronoi diagram of $S$~\cite{LL-KVD-2013}. 
Every degree-3 vertex is equidistant from three
sites and is a point where three Voronoi cells meet.
The number of degree-2 vertices is $\Theta(n)$ for both the geodesic 
nearest-point Voronoi diagram and the geodesic farthest-point Voronoi
diagram~\cite{Aronov-VD-1989,FVD}.
For the geodesic order-$k$ Voronoi diagram, the number of degree-2 vertices
is $O(kn)$~\cite{LL-KVD-2013}, but this bound is not tight.

The first nontrivial algorithm for
computing the geodesic nearest-point Voronoi diagram was given by
Aronov~\cite{Aronov-VD-1989} in 1989, which takes $O((n+m)\log^2
(n+m))$ time.  Later, Papadopoulou and Lee~\cite{PL-VD-1998}
improved the running time to $O((n+m)\log(n+m))$. However,
there has been no  progress since then while the best known
lower bound of the running time remains to be $\Omega(n+m\log m)$.
In fact, Mitchell posed a question whether this gap can be resolved in
the Handbook of Computational Geometry~\cite[Chapter 27]{m-gspno-00}.

For the geodesic order-$k$ Voronoi diagram, the first nontrivial
algorithm was given by Liu and Lee~\cite{LL-KVD-2013} in 2013
for polygonal domains with holes.  Their algorithm works for $m$ point sites
in a polygonal domain with a total of $n$ vertices and takes
$O(k^2(n+m)\log(n+m))$ time.
Thus, this algorithm also works for a simple polygon.
They presented an asymptotically tight combinatorial
complexity of the geodesic order-$k$ Voronoi diagram of $m$ points
in a polygonal domain with
a total of $n$ vertices, which is $\Theta(k(m-k)+kn)$.  However, it is not tight for a
simple polygon: 
the geodesic order-$(m-1)$ Voronoi diagram of $m$ points in a simple
$n$-gon has complexity $\Theta(n+m)$~\cite{FVD}.  There is no bound
better than the one by Liu and Lee known for the complexity of the
geodesic order-$k$ Voronoi diagram in a simple polygon.

For the geodesic farthest-point Voronoi diagram, the first nontrivial
algorithm was given by Aronov et al.~\cite{FVD} in 1993, which takes
$O((n+m)\log(n+m))$ time.  While the best known lower bound is
$\Omega(n+m\log m)$, there has been no progress until Oh et
al.~\cite{fvd_boundary} presented an $O((n+m)\log\log n)$-time
algorithm for the special case that all sites are on the boundary of
the polygon in 2016.  They also claimed that their algorithm
  can be extended to compute the geodesic farthest-point Voronoi
  diagram for any $m$ points contained in a simple $n$-gon in
  $O(n\log\log n+m\log(n+m))$ time. 

\paragraph{Our results.}  Our main contributions are the algorithms
for computing the nearest-point, higher-order and farthest-point
Voronoi diagrams of $m$ point sites in a simple $n$-gon, which improve the
best known ones for $m \leq n/\polylog n$.
To be specific, we present
\begin{itemize}
\item an $O(n+m\log m\log^2 n )$-time algorithm for the geodesic
  nearest-point Voronoi diagram,
\item an $O(k^2m\log m\log^2 n + \min\{nk,n(m-k)\})$-time algorithm
  for the geodesic order-$k$ Voronoi diagram, and
\item an $O(n+m\log m+m\log^2 n)$-time algorithm for the
  geodesic farthest-point Voronoi diagram.
\end{itemize}
Moreover, our algorithms close the gaps of the running times towards
the lower bounds.  Our algorithm for the geodesic nearest-point
Voronoi diagram is optimal for $m\leq n/\log^3 n$.  Since the
algorithm by Papadopoulou and Lee is optimal for $m\geq n$, our
algorithm together with the one by Papadopoulou and Lee gives the
optimal running time for computing the diagram, except for the case
that $n/\log^3 n < m < n$.

Similarly, our algorithm for the geodesic farthest-point Voronoi
diagram is optimal for $m\leq n/\log^2 n$.  Since the algorithm by
Aronov et al.~\cite{FVD} is optimal for $m\geq n$, our algorithm
together with the one by Aronov et al. gives the optimal running time
for computing the diagram, except for the case that $n/\log^2 n<m<n$.
This answers the question posed by Mitchell on the geodesic
nearest-point and farthest-point Voronoi diagrams, except for the
short intervals of $n/\polylog n<m<n$ stated above.

For the geodesic order-$k$ Voronoi diagram, we analyze an
  asymptotically tight combinatorial complexity of the diagram of $m$
  points in a simple $n$-gon, which is
$\Theta(k(m-k)+\min\{nk, n(m-k)\})$.

Other contributions of this paper are the algorithms for computing the
topological structures of the geodesic nearest-point, order-$k$ and
farthest-point Voronoi diagrams which take $O(m\log m\log^2 n)$,
$O(k^2m\log m\log^2 n)$ and $O(m\log m\log^2 n)$ time, respectively.
These algorithms allow us to obtain a dynamic data structure for
answering nearest or farthest point queries efficiently.
In this problem, we are
given a static simple $n$-gon $P$ and a dynamic point set
$S\subseteq P$.  We are allowed to insert points to $S$ and delete
points from $S$.  After processing updates, we are to find the point
of $S$ nearest (or farthest) from a query point efficiently.  This
data structure supports each query in $O(\sqrt m\log (n+m))$ time and
each update in $O(\sqrt m \log m\log^2 n)$ time, where $m$ is the number of
points in $S$ at the moment.


\subsection{Outline}
Our algorithms for computing the geodesic nearest-point, higher-order
and farthest-point Voronoi diagrams are based on a \emph{polygon-sweep
  paradigm}.  For the geodesic nearest-point and higher-order Voronoi
diagrams, we fix a point $o$ on the boundary of the polygon and move
another point $x$ from $o$ in clockwise order along the boundary of
the polygon.  While $x$ moves along the boundary, we compute the
Voronoi diagram of sites contained in the subpolygon bounded by the
shortest path between $o$ and $x$ and the part of the boundary of $P$
from $o$ to $x$ in clockwise order.
For the geodesic farthest-point
Voronoi diagram, we sweep the polygon with a curve consisting of
points equidistant from the geodesic center of the sites.  The curve
moves from the boundary towards the geodesic center.  During the
sweep, we gradually compute the diagram restricted to the region the curve
has swept.

To achieve algorithms running faster than the best known ones for $m\leq n/\polylog n$, 
we first compute
the topological structure of a diagram instead of computing the
diagram itself directly.
The topological structure, which will be defined later,
represents the adjacency of the Voronoi cells and has complexity smaller than
the complexity of the Voronoi diagram.
Once we have the topological structure of a Voronoi diagram, we can 
compute the Voronoi diagram in $O(T_1+T_2\log n)$ time, where
$T_1$ denotes the complexity of the Voronoi diagram and $T_2$ denotes the
complexity of the topological structure of the diagram.

We define four types of events 
where the topological structure of the diagram changes.
To handle each event, we compute a point equidistant from three
points under the geodesic metric. There is no algorithm known
for computing a point equidistant from three points efficiently,
except an $O(n)$-time trivial algorithm.  We present an $O(\log^2 n)$-time
algorithm assuming that the data structure by Guibas and
Hershberger~\cite{shortest-path} is constructed for $P$.
To obtain
this algorithm, we apply two-level binary search on the regions of a
subdivision of the polygon. This algorithm allows us to handle each
event in $O(\polylog \{n,m\})$ time.

One application of the topological structure of a diagram
is a data structure for nearest (or farthest) point queries for a dynamic point set.
To obtain this data structure, we apply the framework given by Bentley and
Saxe~\cite{decomposable} using the algorithm for computing the
topological structure of the geodesic nearest-point (or farthest-point) Voronoi diagram.
We subdivide the dynamic point set into $\sqrt{m}$ almost equal-sized subsets, where $m$
	is the number of the input point.
Then compute the topological structure of the diagram for each subset.
We observe
that we can find the Voronoi cell of each diagram containing a query
point in $O(\log (n+m))$ time once we have the topological structure
of the diagram, which leads to the query time of $O(\sqrt{m}\log (n+m))$.

\section{Preliminaries}
Let $P$ be a simple polygon with $n$ vertices and $S$ be a set of $m$ points
contained in $P$. 
For ease of description, we use $\vd[S]$, $\kvd[S]$ and $\fvd[S]$ (or
simply $\vd$, $\kvd$ and $\fvd$ if they are understood in the context) to
denote the geodesic nearest-point, order-$k$ and farthest-point
Voronoi diagrams of $S$ in $P$, respectively.  
We assume the \emph{general position condition} that no vertex of $P$
is equidistant from two distinct sites of $S$ and no point of $P$ is
equidistant from four distinct sites of $S$. This was also assumed by
in previous work~\cite{Aronov-VD-1989,FVD,LL-KVD-2013,PL-VD-1998} on
geodesic Voronoi diagrams. 


Consider any three points $x,y$ and $z$ in $P$.  We use $\pi(x,y)$ to
denote the shortest path (geodesic path) between $x$ and $y$ contained
in $P$, and $d(x,y)$ to denote the geodesic distance between $x$ and
$y$.  Two geodesic paths $\pi(x,y)$ and $\pi(x,z)$ do not cross each
other, but may overlap with each other.  We call a point $x'$ the
\emph{junction} of $\pi(x,y)$ and $\pi(x,z)$ if $\pi(x,x')$ is the
maximal common path of $\pi(x,y)$ and $\pi(x,z)$.  Refer to
Figure~\ref{fig:center-pseudo}(a).

Consider the set $B$ of points $q$ of $P$ satisfying $d(x,q)=d(y,q)$
  for any two points $x$ and $y$ in $P$. Since $x$ and $y$ are not necessarily contained
  in $S$, the set $B$ may contain a two-dimensional region if there is
a vertex $v$ of $P$ satisfying $d(x,v)=d(y,v)$~\cite{FVD}. However, there are at most
two such two-dimensional regions (including their boundaries) in $B$ and
the other points of $B$ form a simple
curve that is incident to the regions by the general position assumption.
 We call the curve the \emph{bisecting curve}
of $x$ and $y$ and denote it by $b(x,y)$.

Given a point $p\in P$ and a closed set $A\subseteq P$,
we slightly abuse the notation $\pi(p,A)$ to denote the shortest path contained in $P$
connecting $p$ and a point in $A$.
Similarly, we abuse the notation $d(p,A)$ to denote the 
length of $\pi(p,A)$.
It holds that $d(p,A)\leq d(p,q)+d(q,A)$
for any two points $p, q\in P$ and any closed set $A\subseteq P$.

We say a set $A\subseteq P$ is \emph{geodesically
  convex} if $\pi(x,y)\subseteq A$ for any two points $x$ and $y$ in
$A$.  The \emph{geodesic convex hull}
of $S$ is the intersection of all geodesic convex sets in $P$ that
contain $S$.
The geodesic convex hull of a set of $m$ points in $P$ can be computed
in $O(n+m\log(n+m))$ time~\cite{shortest-path}.
The \emph{geodesic center} of a simple polygon $P$ is the point $c\in P$ that
minimizes $\max_{p\in P}d(c,p)$.  The center is
unique~\cite{pollackComputingCenter} and can be computed in $O(n)$
time~\cite{1-center-jnl}.  Similarly, the geodesic center of $S$
can be defined as the point $c\in P$ that minimizes
$\max_{s\in S} d(c,s)$.  It is known that the geodesic center of a set $S$
of $m$ points in $P$
coincides
with the geodesic center of the geodesic convex hull of $S$~\cite{FVD}.
Therefore, we can compute the geodesic center of $S$ by computing the
geodesic convex hull of $S$ and its geodesic center. 
This takes $O(n+m\log(n+m))$ time in total.

\section{Computing the Geodesic Center of Points in a Simple Polygon}
We first present an $O(\log n)$-time algorithm for computing the
geodesic center of three points contained in $P$, assuming that we
have the data structure of Guibas and
Hershberger~\cite{shortest-path,Hersh-shortest-1991}.  
Using ideas from this algorithm, we
present an $O(m\log m\log^2 n)$-time algorithm for computing the geodesic
center of $m$ points in Section~\ref{sec:center-many}.  These
algorithms will be used as subprocedures 
for computing the Voronoi diagrams of points in $P$.

\subsection{Computing the Geodesic Center of Three Points}
\label{sec:center-three}
Let $p_1,p_2$ and $p_3$ be three points in $P$, and let $c$ be the
geodesic center of them.  The geodesic convex hull of $p_1,p_2,p_3$ is
bounded by $\pi(p_1,p_2),\pi(p_2,p_3)$, and $\pi(p_3,p_1)$.  The geodesic convex hull
may have complexity $\Omega(n)$, but its interior is bounded by at most
three concave chains.
This allows us to compute the geodesic center of it efficiently.

We first construct the data structure of Guibas and Hershberger~\cite{shortest-path,Hersh-shortest-1991} 
for $P$ that supports the geodesic distance query
between any two points in $O(\log n)$ time.
To compute $c$, we compute the shortest paths
$\pi(p_1,p_2),\pi(p_2,p_3)$, and $\pi(p_3,p_1)$.  Each shortest path
has a linear size, but we can compute them in $O(\log n)$ time using
the data structure of Guibas and Hershberger. 
Then we find a convex $t$-gon with $t\leq 6$ containing $c$ such
that the geodesic path $\pi(x,p_i)$ has the same combinatorial
structure for any point $x$ in the $t$-gon for each $i=1,2,3$.  To
find such a convex $t$-gon, we apply two-level binary search.  
Then we can compute $c$ directly in constant time inside the $t$-gon.

\paragraph{The data structure given by Guibas and Hershberger.}
Guibas and Hershberger~\cite{shortest-path,Hersh-shortest-1991} gave a data structure of
linear size that enables us to compute 
the geodesic distance between any two query points lying
inside $P$ in $O(\log n)$ time. We call this structure the
\emph{shortest path data structure}. This data structure
can be constructed in $O(n)$ time.

In the preprocessing, they compute a number of shortest paths such
that for any two points $p$ and $q$ in $P$, the shortest path $\pi(p,q)$
consists of $O(\log n)$ subchains of precomputed shortest paths and
$O(\log n)$ additional edges that connect the subchains into one.
In the query algorithm,
they find such subchains and edges connecting them in $O(\log n)$
time.  Then the query algorithm returns the shortest path between two query
points represented as a binary tree of height $O(\log n)$~\cite{Hersh-shortest-1991}.  Therefore,
we can apply binary search on the vertices of the shortest path between any two
points.

\paragraph{Computing the geodesic center of three points: two-level
  binary search.}  Let $\triangle$ be the geodesic convex hull of
$p_1,p_2$ and $p_3$.  The geodesic center $c$ of the three points is
the geodesic center of $\triangle$~\cite{FVD}, thus is contained in
$\triangle$.  If the center lies on the boundary of $\triangle$, we
can compute it in $O(\log n)$ time since it is the midpoint of two
points from $p_1,p_2$ and $p_3$.
So, we assume that the center lies in the interior of
$\triangle$.  Let $p_i'$ be the junction of $\pi(p_i,p_j)$ and
$\pi(p_i,p_k)$ for three distinct indices $i,j$ and $k$ in
$\{1,2,3\}$.  See Figure~\ref{fig:center-pseudo}(a).

\begin{figure}
  \begin{center}
    \includegraphics[width=0.9\textwidth]{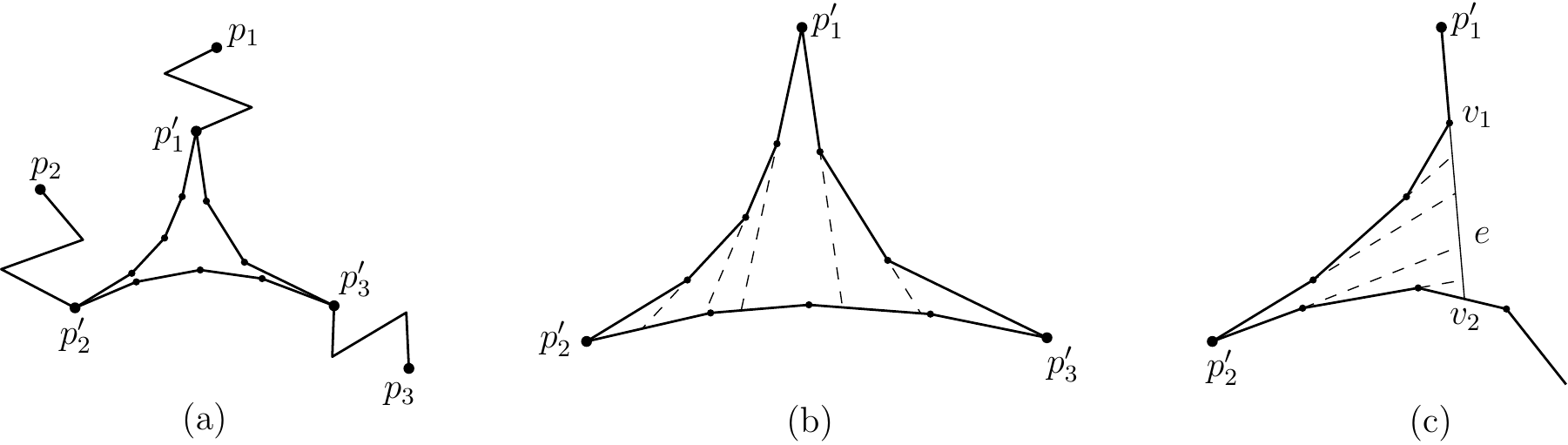}
    \caption {
      (a) $p_i'$ is the junction of $\pi(p_i,p_j)$ and $\pi(p_i,p_k)$
      for three distinct indices $i,j$ and $k$ in $\{1,2,3\}$.
      (b) The subdivision of $\triangle$ with respect to $p_1'$.
      (c) The subdivision of $e$ with respect to $p_2'$.
      \label{fig:center-pseudo}}
  \end{center}
\end{figure}

We use the following lemmas to apply two-level binary search.
Recall that we have already constructed the shortest path data structure for $P$.
\begin{lemma}[\cite{shortest-path}]
  \label{lem:last-common-vertex}
  We can compute the junctions $p_1',p_2'$ and $p_3'$ in $O(\log n)$ time.
\end{lemma}

\begin{lemma}[\cite{ray_shooting}]
  \label{lem:extension}
  Given a point $p \in \triangle$ and a direction, we can find the
  first intersection point of the boundary of $\triangle$ with the ray
  from $p$ in the direction in $O(\log n)$ time.
\end{lemma}	
\begin{proof}
Chazelle et al.~\cite{ray_shooting} showed that the first
intersection point can be found $O(\log n)$ time
if we have balanced binary search trees representing the maximal
concave curves lying on the boundary of $\triangle$.
We can obtain such balanced binary search trees from the shortest path
data structure in $O(\log n)$ time.
\end{proof}
%

\paragraph{The first level.}  Imagine that we subdivide $\triangle$
into $O(n)$ regions with respect to $p_1'$ by extending the edges of
$\pi(p_1',p_2')\cup\pi(p_1',p_3')$ towards $\pi(p_2',p_3')$.  See
Figure~\ref{fig:center-pseudo}(b).  The extensions of the edges can be
sorted in the order of their endpoints appearing along $\pi(p_2',p_3')$.
Consider the
subdivision of $\triangle$ by the extensions, and assume that we can
determine which side of a given extension in $\triangle$ contains $c$ in $T(n)$ time.
Then we can compute the region of the subdivision containing $c$ in
$O(T(n)\log n)$ time by applying binary search on the extensions.  Note that
any point $x$ in the same region has the same combinatorial structure of
$\pi(x,p_1)$ (and $\pi(x,p_1')$).

We also do this for $p_2'$ and $p_3'$. Then we have three
regions whose intersection contains $c$.  Let $D$ be the
intersection of these three regions.  We can find $D$ in
constant time by the following lemma.
\begin{lemma}
The intersection $D$ is a convex polygon with at most six
edges from extensions of the regions.
\end{lemma}
\begin{proof}
	Let $C_i$ be the region of the subdivision with respect to
	$p_i'$ containing $c$ for each $i=1,2,3$.  The boundary of $C_i$ consists
	of two extensions and a part of $\pi(p_j', p_k')$, where
	$j$ and $k$ are distinct indices in
	$\{1,2,3\}\setminus\{i\}$.  This means that $C_i$ does not
	contain any concave curve which comes from $\pi(p_i',p_j')$ or
	$\pi(p_i',p_k')$.  Therefore, the intersection of the
	three regions is a convex polygon with at most six edges.
\end{proof}

We do not subdivide $\triangle$ 
explicitly.  Because we have $\pi(p_1',p_2')$ and
$\pi(p_1',p_3')$ in binary trees of height $O(\log n)$, we
can apply binary search on the extensions of the edges of the
geodesic paths 
without subdividing $\triangle$ explicitly.
In this case, during the binary search, 
we compute the extension of a given edge of $\pi(p_1',p_2')\cup\pi(p_1',p_3')$
using Lemma~\ref{lem:extension}, which takes $O(\log n)$ time.

There is a vertex $p$ on the boundary of $\triangle$
such that for any point $x$ contained in $D$ we have
$d(p_1,x)=d(p_1,p)+\|p-x\|$, where $\|p-x\|$ is the
Euclidean distance between $p$ and $x$.
Moreover, we
already have $p$ from the computation of the region
containing $c$ in the subdivision with respect to $p_1'$.
The same holds for $p_2$ and $p_3$.
Therefore, we can compute the point $c$ that minimizes the
maximum of $d(c,p_1), d(c,p_2)$ and $d(c,p_3)$ in constant
time inside $D$.

Therefore, we have the following lemma.
\begin{lemma}
  Assuming that we can determine which side of
  an extension in $\triangle$ contains $c$ in $T(n)$
  time, we can compute the geodesic center $c$ in $O((T(n)+\log n)\log
  n)$ time.
\end{lemma}

\paragraph{The second level.}  In the second level binary
search, we determine which side of an extension $e$ 
in $\triangle$ contains $c$.  Without loss of
generality, we assume that $e$ comes from the subdivision
with respect to $p_1'$.  Then $\pi(p_1,x)$ has
the same combinatorial structure for any point $x \in e$.

This subproblem was also considered in a few previous works
on computing the geodesic center of a
simple polygon~\cite{1-center-jnl,pollackComputingCenter}.
They first compute the point
$c_e$ in $e$ that minimizes $\max_{p\in P}d(p,c_e)$, that
is, the geodesic center of the polygon restricted to
$e$.  By using $c_e$ and its farthest point, 
Pollack et al.~\cite{pollackComputingCenter}
presented a way to decide which side of $e$ contains the
geodesic center of the polygon in constant time.  However,
to compute $c_e$, they spend $O(n)$
time.

In our problem, we can do this in $O(\log n)$ time using the
fact that the interior of $\triangle$ is bounded by at most three
concave chains.  By
this fact, there are two possible cases: $c_e$ is an
endpoint of $e$, or $c_e$ is equidistant from $p_1$ and
$p_i$ for $i=2$ or $3$. 
We compute the point on $e$ equidistant from $p_1$ and $p_2$, and 
the point on $e$ equidistant from $p_1$ and $p_3$.
Then we find the point $c_e$ among the two points and the two endpoints of $e$.
In the following, we show how to
compute the point on $e$ equidistant from $p_1$ and $p_2$ if
it exists.  The point on $e$ equidistant from
$p_1$ and $p_3$ can be computed analogously.
 
Observe that $e$ can be subdivided into $O(n)$ disjoint line
segments by the extensions of the edges of 
$\pi(p_2',v_1)\cup\pi(p_2',v_2)$ towards $e$,
where $v_1$ and $v_2$ are endpoints of $e$.  See
Figure~\ref{fig:center-pseudo}(c).
For any point $x$ in the same line segment,
$\pi(p_2,x)$ 
has the same combinatorial structure.

We first claim that there is at most one point on $e$
equidistant from $p_1$ and $p_2$.  Assume to the contrary
that there are two such points $x$ and $y$.  Without loss of
generality, we assume that $d(x,p_1)<d(y,p_1)$.  By
definition, $d(x,p_1)=d(x,p_2)$ and $d(y,p_1)=d(y,p_2)$.  By
the construction of $e$, $d(y,p_1)=d(x,p_1)+d(x,y)=d(x,p_2)+d(x,y)$, but
$d(y,p_2)<d(x,p_2)+d(x,y)$ by triangle inequality and the construction
of $e$. This contradicts that $d(y,p_1)=d(y,p_2)$.

Thus we can apply binary search on the line segments in the subdivision
of $e$.  As we
did before, we do not
subdivide $e$ explicitly. Instead, we use the binary
trees of height $O(\log n)$ representing $\pi(p_2',v_1)$ and $\pi(p_2',v_2)$.
For a point $x$ in $e$, by comparing
$d(p_1,x)$ and $d(p_2,x)$, we can determine which part of
$x$ on $e$ contains the point equidistant from
$p_1$ and $p_2$ in constant time.  
In this case, we can compute the extension from an edge towards $e$
in constant time since $e$ is a line segment.
Thus, we complete the
binary search in $O(\log n)$ time.

Therefore, we can compute $c_e$ in $O(\log n)$ time and
determine which side of $e$ in $\triangle$ contains $c$ in the same time
using the method of Pollack et al~\cite{pollackComputingCenter}.
The following lemma summarizes this section.
\begin{lemma}
  Given any three points $p_1,p_2$ and $p_3$ contained in a simple $n$-gon $P$,
  the geodesic center of $p_1,p_2$ and $p_3$ can be
  computed in $O(\log^2 n)$ time after the shortest
  path data structure for $P$ is constructed in 
  linear time.
\end{lemma}

\paragraph{Remark.}
The observations in this section together with the tentative prune and search technique~\cite{tentative-1993} yield an $O(\log n)$-time algorithm for computing the
geodesic center of any three points. Refer to~\cite[Section 3.4]{tentative-1993}.
But the running time for computing the geodesic center of three points is subsumed by
the overall running times for computing the Voronoi diagrams since 
the algorithms in Sections~\ref{sec:threepoints-equidistant} and~\ref{sec:point-line} 
take $O(\log^2 n)$ time.
It seems unclear whether these algorithms can be improved
to $O(\log n)$ time by applying this technique.
Thus we do not provide
details of the $O(\log n)$-time algorithm for this problem here.

\subsubsection{The Point Equidistant from Three Points}
\label{sec:threepoints-equidistant}
The geodesic center of three points in a simple polygon may not be 
equidistant from all of them.
Moreover, three points in a simple polygon may have no point equidistant
from them in the polygon.
For example, any three points whose geodesic convex hull is an
obtuse triangle contained in a simple polygon have their geodesic center
at the midpoint of the longest side of the obtuse triangle,
but it is not equidistant from the three points.
If the three points are almost aligned along a line, they may have no equidistant
point in the polygon.

Under the general position condition on the sites, 
there is at most one point equidistant from three sites. 
However, for three points which are not necessarily in $S$, 
there may be an infinite number of points equidistant from the three points.
In this case, we compute the one closest to the three points.

We can compute the closest equidistant point from any three points 
efficiently using the algorithm for computing the center of them
if any equidistant point exists.
\begin{lemma}
  \label{lem:equidistant}
  Given any three points in a simple polygon with $n$ vertices, 
  we can compute the closest
  equidistant point from them in $O(\log^2 n)$ time if it exists.
\end{lemma}
\begin{proof}
  Let $p_1,p_2$ and $p_3$ be three input points and $c^*$ be the closest
  equidistant point from them.  We first compute the geodesic center $c$ of
  the three points in $O(\log^2 n)$ time. If $c$ is equidistant from the three
  points, we are done.
Otherwise, it is equidistant from
  only two of them, say $p_2$ and $p_3$, and it lies on $\pi(p_2,p_3)$.
  Then we have $d(p_1,c)<d(p_2,c)=d(p_3,c)$.
  Recall that $p_i'$ denote the junction of $\pi(p_i,p_j)$ and $\pi(p_i,p_k)$
  for three distinct indices $i, j$ and $k$ in $\{1,2,3\}$.
  If $c$ lies on $\pi(p_2,p_2')$, 
  $P$ has no point equidistant from the three points because $d(p_1,p_2')<d(p_3,p_2')$,
  and therefore $d(p_1,x)<d(p_2,x)=d(p_3,x)$ for
  any point $x$ on the bisecting curve of $p_2$ and $p_3$.
  Similarly, if $c$ lies on $\pi(p_3,p_3')$,
  $P$ has no point equidistant from the three points.
  So we assume that $c$ lies in an edge
  $v_2v_3$ of $\pi(p_2',p_3')$ 
  with $d(p_2',v_2)<d(p_2',v_3)$. Then $v_2v_3$ subdivides $P$ into two subpolygons.
  See Figure~\ref{fig:equidistant}(a).

  \begin{figure}
    \begin{center}
       \includegraphics[width=0.8\textwidth]{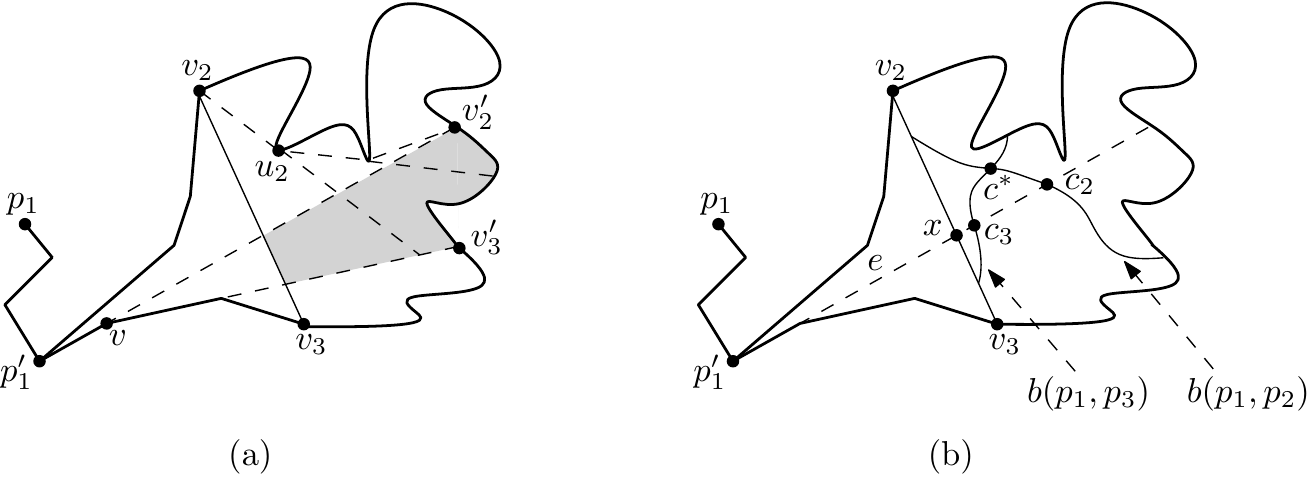}
      \caption{\label{fig:equidistant} (a) The region (the gray region) 
        in the subdivision
        with respect to $p_1'$ containing $c^*$ is subdivided with respect to $p_2$.
        (b) $c^*$ is the intersection point between $b(p_1,p_2)$ and 
        $b(p_1,p_3)$.}
    \end{center}
  \end{figure}

  Assume that $c^*$ exists in $P$. Then it lies in the subpolygon $P'$ of $P$
  bounded by $v_2v_3$ and not containing $p_1$.
  We claim that the part of $\pi(p_1,c^*)$ contained in $P'$ is just a line
  segment.
  Assume to the contrary that $\pi(p_1,c^*)\cap P'$ is a chain of at least two
  line segments.
  Then each point $q$ where the chain makes a turn is a vertex of $P'$ other than
  $v_2$ and $v_3$. 
  Then $\pi(p_2,c^*)$ or $\pi(p_3,c^*)$ also passes through $q$ because $q$ is a
  vertex of $P$ and $\pi(p_i,c^*)$ do not cross each other for $i=1,2,3$. 
  This is a contradiction. To see this, observe that
  $c^*$ is an intersection point of $b(p_1,p_2)$, $b(p_2,p_3)$ and $b(p_3,p_1)$.
  Thus, 
  for any point $x$ in $b(p_1,p_2)$ (or $b(p_3,p_1)$), there is no vertex of $P$
  where both $\pi(p_1,x)$ and  $\pi(p_2,x)$ (or $\pi(p_3,x)$) make a turn.
  Therefore, the part of $\pi(p_1,c^*)$ contained in $P'$ is a line segment.

  
  Consider the subdivision of $P'$ by the extensions
  of the edges of $\pi(p_1',v_2)\cup\pi(p_1',v_3)$ with respect to $p_1'$.
  Figure~\ref{fig:equidistant}(a) shows a region (gray)
  in such a subdivision. 
  We find the region $C$ containing $c^*$ by applying binary search on the extensions.
  For an extension $e$, 
  we determine which side of $e$ in $P'$ contains $c^*$
  as follows. 
  There is at most one point on $e$ equidistant from $p_1$ and $p_i$ for $i=2$ or $3$.
  Let $c_i$ be the intersection point of $b(p_1,p_i)$ with $e$.
  Figure~\ref{fig:equidistant}(b) shows $b(p_1,p_2)$ 
  and $b(p_1,p_3)$ contained in $P'$. Each bisecting curve is a simple
  connected curve and it intersects $e$ at most once by the construction
  of $e$. 
  Moreover, $b(p_1,p_i)$ intersects $\pi(p_i,x)$ exactly once
  for $i=2,3$, where $x$ is the 
  intersection point of $v_2v_3$ with $e$.
  And $c^*$ is an intersection point of $b(p_1,p_2)$ and $b(p_1,p_3)$.
  Therefore, if $b(p_1,p_i)$ does not intersect $e$, the point 
  $c^*$ lies in the side of $e$
  containing $v_i$ for $i=2,3$.
  If $c_2$ comes after $c_3$ along $e$ from $p_1'$, the point
  $c^*$ lies in the side of $e$ containing $v_2$ .
  Otherwise, $c^*$ lies in the side of $e$ containing $v_3$.
  Thus we can determine which part of $e$ inside $P'$ contains
  $c^*$ in constant time after computing $c_2$ and $c_3$ in $O(\log n)$.
  To compute the extension from an edge, we can use the ray-shooting 
  algorithm which takes $O(\log n)$
  time~\cite{ray_shooting} since the endpoints of the extensions
  lie on the boundary of $P$.
  Therefore, we can find the region $C$ containing $c^*$ in the subdivision
  of $P'$ with respect to $p_1'$ in $O(\log^2 n)$ time.
 
  We let $vv_2'$ and $vv_3'$ be two extensions bounding $C$.
  See the gray region in Figure~\ref{fig:equidistant}(a).
  By applying binary search on $\pi(v_2,v_2')$,
  we find the junction $u_2$ of $\pi(p_2',v_2')$ and $\pi(p_2',c^*)$
  in $O(\log^2 n)$ time as follows.  
  We find the point $c_2'$ on the extension of an edge on $\pi(v_2,v_2')$ that
  is equidistant from $p_1$ and $p_2$
  in $O(\log n)$ time.  Then we compare $d(c_2',p_1)$ and $d(c_2',p_3)$, which
  determines whether the junction $u_2$ lies before the edge from $p_2'$
  along $\pi(p_2',v_2')$
  or not.
  We also do this for $p_3$ and find the junction $u_3$
  of $\pi(p_3',v_3')$ and $\pi(p_3',c^*)$ in $O(\log^2 n)$ time.
  
  Now we have 
  the junction $u_i$ of $\pi(p_i,v_i')$ and $\pi(p_i,c^*)$ 
  for all $i=2, 3$. 
  We observe that the geodesic path $\pi(p_i,c^*)$ is the concatenation of
  $\pi(p_i,u_i)$ and the line segment $u_ic^*$.
  Therefore $d(p_j,\cdot)$ for $j=1,2,3$ is a hyperbolic function with 
  some domain $D$ containing $c^*$.
  We compute the three hyperbolic functions and find points
  where the three hyperbolic functions have the same value without considering $D$.
  Since we do not consider $D$, some point that we compute may not be equidistant
  from $p_1,p_2$ and $p_3$. We check additionally if each such point is equidistant
  from the three points. In this way, we can compute $c^*$ in $O(\log^2 n)$ time in total.
\end{proof}

\subsubsection{The Point Equidistant from Two Points and a Line Segment}
\label{sec:point-line}
Using a way similar to the one in Section~\ref{sec:threepoints-equidistant}, 
we can compute the point equidistant from two
points in $P$ and a line segment contained in $P$ under the geodesic
metric.  
Given two points and a line segment contained in $P$,
  if more than one point of $P$ are equidistant from them,
we choose the one closest to them.
This will be also used as a
subprocedure for computing Voronoi diagrams.

\begin{lemma}
  \label{lem:equidistant-line}
  Given any two points and any line segment contained in a simple $n$-gon $P$,
  we can compute the closest equidistant point from them
  under the geodesic metric in $O(\log^2 n)$ time if it exists.
\end{lemma}
\begin{proof}
  Let $p_1$ and $p_2$ be any two points in $P$, and $t_1t_2$ be any line
  segment contained in $P$.  
  We compute two points $t_1'$ and $t_2'$ on $t_1t_2$
  that are closest to $p_1$ and $p_2$ in $O(\log n)$ time, respectively. 
  We compute the junction $p_1'$ of $\pi(p_1,p_2)$ and $\pi(p_1,t_1')$,
  and the junction $p_2'$ of $\pi(p_2,p_1)$ and $\pi(p_2,t_2')$
  in $O(\log n)$ time.
  Without loss of generality, we assume
  that $t_1', t_2', p_2'$ and $p_1'$ appear in order along the boundary of the
  geodesic convex hull $\ch$ of them as shown in 
  Figure~\ref{fig:center-line-point}. 
  Note that the interior of $\ch$ is connected.

  \begin{figure}
    \begin{center}
      \includegraphics[width=0.9\textwidth]{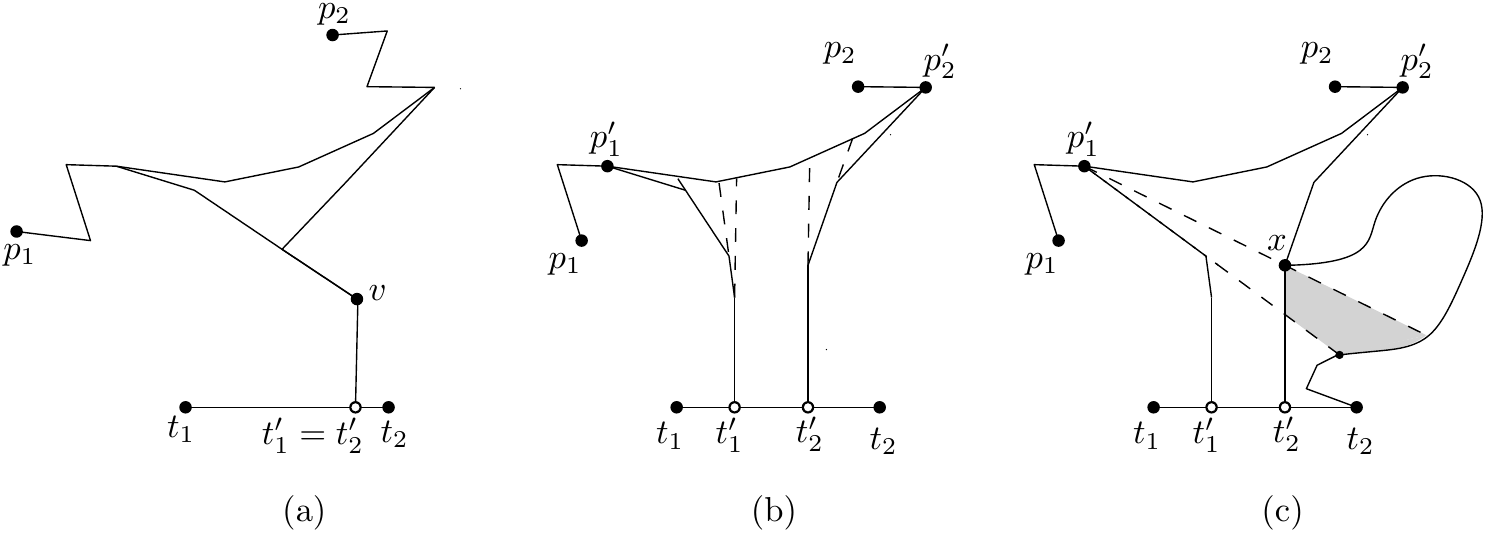}
	  \caption{
		(a) $t_1'$ may coincide with $t_2'$.
		(b) The subdivision of $\ch$ with respect to $t_1't_2'$.
		(c) If $c$ lies on $\pi(p_2,t_2')$, we first consider the
		subdivision $P'$ with respect to $p_1'$ and find the region
		containing $c^*$ (the gray region).
		Then consider the subdivisions of the region with respect to
		$t_1t_2$ and $p_2'$, find the regions containing $c$,
		and compute $c$ directly inside their intersection.
		\label{fig:center-line-point}}
     \end{center}
   \end{figure}

   Let $c^*$ be the closest equidistant point 
   from $p_1$, $p_2$ and $t_1t_2$ in $P$. 
   We first compute the geodesic center $c$ of them in $O(\log^2n)$ time
   as follows. 
   If $c$ lies on the boundary of $\ch$, we can compute it in $O(\log n)$
   time since it is the midpoint of two points from $p_1, p_2, t_1'$ and $t_2'$.
   Thus we assume that $c$ lies in the interior of $\ch$.

   Consider the subdivision of $\ch$ with respect to $p_1'$
   by the extensions of edges in $\pi(p_1',t_1')\cup\pi(p_1',p_2')$
   in direction opposite to $p_1'$. We can determine
   which side of a given extension contains $c$ in $O(\log n)$ time. Therefore, we
   can compute the region of the subdivision containing $c$ in $O(\log^2n)$ time without
   constructing the subdivision as we do for the first level binary search in 
   Section~\ref{sec:center-three}. Similarly, we can compute the region of the subdivision
   of $\ch$ containing $c$ with respect to $p_2'$ by the  extensions of edges in 
   $\pi(p_2',t_2')\cup\pi(p_2',p_1')$ in direction opposite to $p_2'$ 
   in the same time.
   Now consider the subdivision of $\ch$ by the extensions of edges
   in $\pi(t_1',p_1')\cup\pi(t_2',p_2')$ in direction opposite to $t_1'$ and $t_2'$.
   See Figure~\ref{fig:center-line-point}(b).
   We can compute the region
   of the subdivision containing $c$ in the same time. Then the intersection of the three 
   regions is a convex polygon with at most six vertices and we can find the geodesic center $c$ in the intersection in constant time.

   If $c$ is equidistant from $p_1$, $p_2$ and $t_1t_2$, we are done.
   Otherwise, it is equidistant from only two of them and it lies in the shortest path 
   connecting the two in $P$. Let $xy$ be the edge of the shortest path that contains $c$.
   If $xy$ is not on the boundary of the interior of $\ch$, then 
   $P$ has no point equidistant from $p_1$, $p_2$ and $t_1t_2$ by an argument similar
   to the one in the first paragraph of the proof of Lemma~\ref{lem:equidistant}. 
 
   Assume that $xy$ is on the boundary of the interior of $\ch$. Then $xy$ subdivides 
   $P$ into two subpolygons and $c^*$, if it exists in $P$, lies in the subpolygon $P'$ 
   not containing the interior of $\ch$.
   In the case that $xy$ lies on $\pi(p_1,p_2)$,
   consider the subdivision of $P'$ with respect to $t_1't_2'$ by the extensions of
   edges in $\pi(t_1',x)\cup\pi(t_2',y)$ towards $P'$.
   For the other case, that is, $xy$ lies on $\pi(p_i,t_i')$ for $i=1$ or $2$,
   consider the subdivision of $P'$ with respect to $p_j'$ by
   the extension of edges in $\pi(p_j',x)\cup\pi(p_j',y)$ for
   $j\in\{1,2\}\setminus\{i\}$. For an extension $e$, we determine which
   side of $e$ in $P'$ contains $c^*$ as follows. We compute the intersection point
   $q$ of $b(p_1,p_2)$ with $e$ by performing binary search
   on the intersections of $e$ by the extensions of edges on the shortest path
   from $p_2$ to an endpoint of $e$. Then we determine which side of $e$
   contains $c$ by comparing the geodesic distances $d(p_1,q)$ and $d(t_1t_2,q)$.
   This can be done in $O(\log n)$ time as the intersection point $q$ and
   the distances  can be computed in $O(\log n)$ time.
 See Figure~\ref{fig:center-line-point}(c).
   In any case, as we do
   in the proof of Lemma~\ref{lem:equidistant}, we can find $c^*$ in $P'$ in 
   $O(\log^2 n)$ time if it exists.
\end{proof}

\subsection{The Geodesic Center of Points in a Simple Polygon}
\label{sec:center-many}
Combining the result in the previous subsection with the
algorithms for computing the center of points in the
  plane~\cite{m-ltalprp-83} and for computing the geodesic center of a
  simple polygon~\cite{pollackComputingCenter},
we can compute the geodesic center of a set of $m$ points contained in
a simple polygon $P$ with $n$ vertices in 
$O(m\log m\log^2 n)$ time after computing the shortest path data
structure for $P$.  This algorithm will be used as a subprocedure for
computing the topological structure of the geodesic farthest-point
Voronoi diagram of points in $P$.

To compute the center of a simple polygon, Pollack et
al.~\cite{pollackComputingCenter} first triangulate the input polygon
and construct a balanced binary search tree on the chords of the
triangulation using the algorithm by Guibas et
al~\cite{shortest-path-tree}.  Then, they find the triangle $t$ of
the triangulation containing the geodesic center by applying binary
search with a chord-oracle to the balanced binary search tree.  A
chord-oracle is the procedure for determining which side of a given
chord contains the geodesic center.  They subdivide $t$ further and
locate a smaller triangle $t'$ such that the geodesic path
from a vertex of $P$ to any point in $t'$ has the same
  combinatorial structure.  Finally, they find the center inside $t'$
which is the lowest point of the upper envelope of some distance
functions within domain $t'$.

To remove the linear dependency of $n$ in the time complexity of our
algorithm, we again use the shortest path data structure.  We may
assume that we already have the triangulation of $P$ and the balanced
binary search tree on the chords because the algorithm for
constructing the shortest path data structure constructs them.

\paragraph{Computing the triangle of the triangulation containing
  the geodesic center.}  Let $c$ be the geodesic center of $S$.  We
apply a chord-oracle described in Lemma~\ref{lem:chord-oracle} to
$O(\log n)$ chords of the triangulation and find the triangle $t$
containing $c$ in $O(m\log^2 n)$ time.

\begin{lemma}
  \label{lem:chord-oracle}
  Given a chord of $P$, we can determine which side of the chord
  contains the geodesic center in $O(m\log n)$ time.
\end{lemma}
\begin{proof}
  Let $e$ be a given chord.  We will find the point $c_e$ on $e$ that
  minimizes the geodesic distance from its farthest site in $S$.
  Based on $c_e$ and its farthest sites in $S$, we can determine which
  side of $e$ contains $c$ in constant time using the method by
  Pollack et al~\cite{pollackComputingCenter}.

  To compute $c_e$, we do the followings.  The chord subdivides $P$
  into two regions.  Let $S_1$ be the set of sites in $S$ contained in
  one region of $P$ and $S_2$ be the set of sites in $S$ contained in
  the other region of $P$.  We pair the sites in the same set $S_i$
  for $i=1,2$.  Then we have $\lfloor m/2 \rfloor +1$ pairs.  For any
  pair $(s,s')$ we have, the bisecting curve of $s$ and $s'$
  intersects $e$ at most once because both $s$ and $s'$ lie in the
  same side of $e$.  We can prune and search on $S$ with this
  property.

  We compute the intersection $p$ of the bisecting curve of $s$ and
  $s'$ with $e$ for each pair $(s,s')$ 
  in $O(\log n )$ time as follows. Observe that $e$ is subdivided into 
  $O(n)$ disjoint line segments by the extensions  with respect to $s$ and by
  the extensions with respect to $s'$.
  We observe that $p$ lies between $x$ and $x'$, where $x$ and $x'$
  are the points on $e$ that are closest to $s$ and $s'$, respectively.
  We can compute $x$ and $x'$ in $O(\log n)$ time.
  We apply binary search on the line segments in the
  subdivision lying between $x$ and $x'$ to find $p$ as follows.
  Initially, the search space for $s$ is the set of the extensions
  of edges of $\pi(s,x)\cup\pi(s,x')$ towards $xx'$,
  and the search space for $s'$ is the set of the extensions of 
  edges of $\pi(s',x)\cup\pi(s',x')$ towards $xx'$.
  For each iteration of the binary search, we choose the median of each search space,
  and let $q$ and $q'$ be the intersection of $xx'$ with the two medians.
  We can choose one line segment which does not
  contain $p$ among $xq$, $xq'$, $x'q$ and $x'q'$ in constant time
  based on distances $d(s,q)$ and $d(s',q')$. This is because one of $d(s,\cdot)$ and $d(s',\cdot)$ is
  increasing and the other is decreasing in the domain of $xx'$. 
  Since at least one of the two search spaces is reduced by a constant fraction in every
  iteration, we can find $p$ in $O(\log n)$ time.

We find the median of the $\lfloor m/2 \rfloor +1$ intersections by the  bisecting curves on
$e$ and determine which
  side of the median contains $c_e$ in $m\cdot O(\log n)$ time using
  the convexity of $d(s,\cdot)$ on $e$ for each site $s$ in $S$.  Let
  $e'$ be the side (line segment) of the median in $e$ that does not
  contain $c_e$.  There are at least $\lfloor m/4\rfloor$ pairs of sites whose
  intersection points lie on $e'$.  For such a pair $(s,s')$, we have
  either $d(s,z) > d(s',z)$ or $d(s,z)< d(s',z)$ for every point
  $z \in e\setminus e'$. Thus we can discard at least one site for each such pair:
  discard $s$ if $d(s,z)<d(s',z)$, and discard $s'$ otherwise.
  Therefore, we can discard at least $\lfloor m/4\rfloor$ sites in each iteration.

  After discarding at least $\lfloor m/4\rfloor$ sites, we update $S_1$ and
  $S_2$ accordingly.  Then we pair the remaining sites in the same set and
  discard some sites repeatedly until there remain only a constant
  number of candidates of the sites in $S$ farthest from $c_e$. This
  takes $O(m\log n)$ time because the number of sites in $S$ we
  consider decreases by a constant factor.

  For a constant number of candidates, we compute $c_e$ directly in
  $O(\log n)$ time as we do in Section~\ref{sec:center-three}.
\end{proof}

\paragraph{Finding a small triangle containing the geodesic
  center.}  Now we have the triangle $t$ of the triangulation
containing $c$.  For each site $s$, we consider the geodesic convex
hull $\ch_s$ of $s\cup t$.  Imagine the subdivision of $t$ by the
extensions of the edges of $\ch_s$, except the edges of $t$, towards $t$.
See Figure~\ref{fig:subdivide-ch}(a).  Every region in this
subdivision is a convex polygon with at most five vertices.  These regions
can be sorted along the boundary of $t$.

Applying binary search, we can find the region in the subdivision of $t$
with respect to $s$ containing $c$ in $O(m\log^2 n)$ time.  However,
in this case, we have $m$ sites. If we do this for every site, it
takes $O(m^2\log^2 n)$ time.  To do this more efficiently, we apply
additional prune and search.

\begin{figure}
  \begin{center}
    \includegraphics[width=0.7\textwidth]{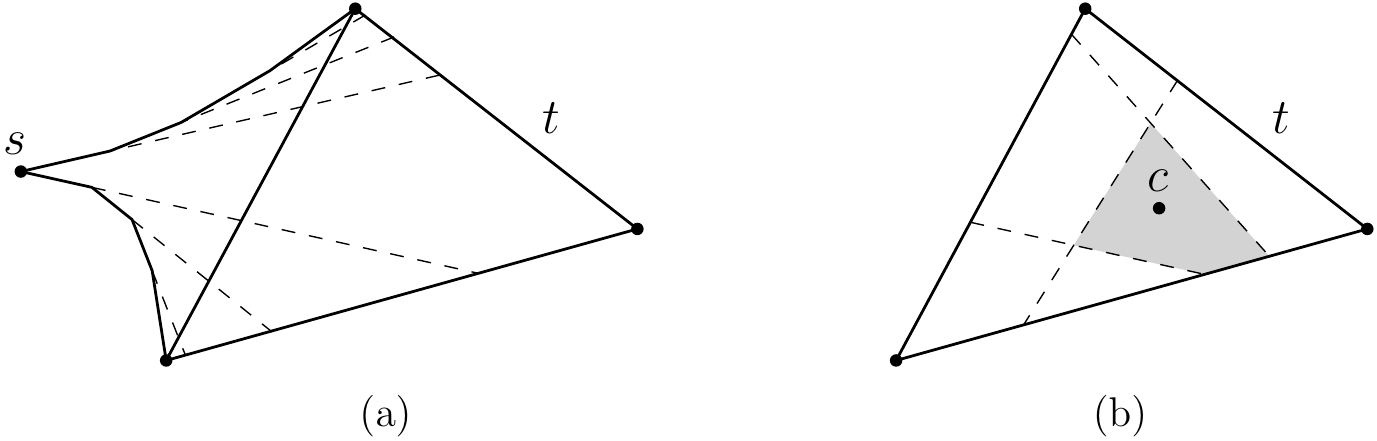}
    \caption {
      (a) The subdivision of $t$ with respect to a site $s$.
      (b) The arrangement of extensions from the subdivisions of $t$
      with respect to the sites in $S$.
      \label{fig:subdivide-ch}}
  \end{center}
\end{figure}

For each site $s$, we choose the median extension in the
subdivision of $t$ with respect to $s$.  Let $L$ be the set
of all median extensions from all sites. See
Figure~\ref{fig:subdivide-ch}(b). 
For a site $s$, the search space (regions in the subdivision of $t$
with respect to $s$) can be halved by determining
which side of the median extension for $s$ in $L$ contains $c$.
Once we have the region containing $c$ in the arrangement of $L$,
we can determine which side of the median extension for $s$ for every site $s$.
We find the region containing $c$ in the arrangement of $L$ using a $(1/r)$-\emph{cutting}
for $L$ as follows.

Consider the range space
$$X=(L, \{\{\ell \in L : \ell \cap \tau \neq \phi\} : 
\tau\subseteq t \textnormal{ is a triangle.}\}).$$
Then $X$ has finite VC-dimension, which can be shown by using a way
similar to the one for lines in the plane~\cite{nets}.  Let $r$ be a
sufficiently large constant.  We compute an $(1/r)$-net $N\subseteq L$
for $X$ of size $O(r\log r)$.  Then any triangle in the triangulation
of $N$ intersects $O(|L|/r)$ line segments by the property of
$(1/r)$-nets, where $|L|$ denotes the cardinality of $L$.  
Note that every line segment in $N$ is a chord of $P$.
By applying the chord-oracle for every line segment in
$N$, we can compute the triangle $t'$ in the triangulation of $N$
containing $c$ in $O(m\log n)$ time.  We search further with the line
segments in $L$ intersecting $t'$.  In $O(\log m)$ iterations, we can
find a triangle $\tau\subseteq t'$ containing $c$ which is intersected
by a constant number of line segments in $L$.  Then we directly
compute the region of the arrangement of $L$ containing $c$ in
$O(\log n)$ time. This takes $O(m\log m\log n)$ time in total.

Therefore, in $O(m\log m\log n)$ time, the search space is halved
(regions in the subdivision of $t$) for every site $s$.  Thus, in
$O(m\log m\log^2 n)$ time, we can find the region containing $c$ in the
subdivision of $t$ with respect to $s$ for every site $s$.

We can compute the intersection of all the regions containing $c$
in $O(m\log m)$ time because each region is a convex polygon with at most five vertices.
After computing the intersection, we triangulate the intersection and
find the triangle containing $c$ in $O(m\log m\log n)$ time.  The
resulting triangle $\tau'$ satisfies the condition we want to obtain.

\paragraph{Finding the geodesic center inside the smaller
  triangle.}  For each site $s$, the geodesic path $\pi(s,x)$ between
$s$ and any point $x$ in $\tau'$ has the same combinatorial structure,
which means that $\pi(s,x)$ is a hyperbolic function.  Thus, our
problem reduces to computing the lowest point of the upper envelope of
$O(m)$ functions.  Pollack et al.~\cite{pollackComputingCenter}
present a procedure for this problem that takes $O(m)$ time.

Therefore, we have the following theorem.
\begin{theorem}
  \label{thm:center-point}
  The geodesic center of $m$ points contained in a simple polygon with
  $n$ vertices can be computed in $O(m\log m\log^2 n)$ time after the
  shortest path data structure for the simple polygon is constructed.
\end{theorem}

\section{Topological Structures of Voronoi Diagrams}
In this section, we define the topological structure of a Voronoi
diagram and show how to compute the Voronoi diagram from its
topological structure.  The topological structure of a Voronoi
diagram represents the adjacency of the Voronoi cells.

The common boundary of any two adjacent Voronoi cells is connected for
the nearest-point and farthest-point Voronoi diagrams of point sites
in a simple polygon~\cite{Aronov-VD-1989,FVD}.  Similarly, 
this also holds for the higher-order Voronoi diagram of
point sites in a simple polygon.

\begin{figure}
  \begin{center}
    \includegraphics[width=0.7\textwidth]{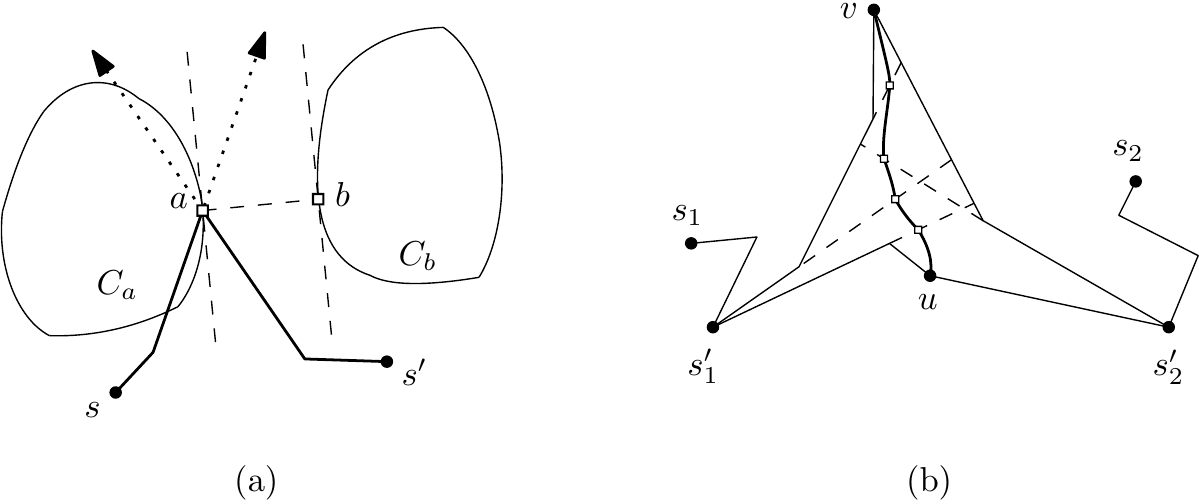}
    \caption {
      (a) For any point $x$ in $\pi(a,b)$,
      we have $d(x,s)\geq d(x,s')$. This contradicts that $C_a$ and $C_b$ are in the same
      Voronoi cell of $\kvd$.
      (b) The number of edges in the common boundary of two adjacent Voronoi 
      cells is bounded by the total complexity of
      $\pi(s_1',v)$, $\pi(s_1',u)$, $\pi(s_2',v)$ and $\pi(s_2',u)$,
      where $u$ and $v$ are endpoints of the common boundary.
      \label{fig:bisect}}
  \end{center}
\end{figure}

\begin{lemma}
  A Voronoi cell of $\kvd$ is connected for any $k$ with
  $1\leq k\leq m-1$.  Moreover, the common boundary of any two
  adjacent Voronoi cells in $\kvd$ is connected.
\end{lemma}
\begin{proof}
  Assume to the contrary that the Voronoi cell of a $k$-tuple is not
  connected.  Consider two connected components $C_a$ and $C_b$ of the
  Voronoi cell, and let $a$ and $b$ be the two points with $a\in C_a$
  and $b\in C_b$ that minimize $d(a,b)$.  Note that the line tangent
  to $C_a$ at $a$ is orthogonal to the edge of $\pi(a,b)$ incident to
  $a$.  See Figure~\ref{fig:bisect}(a).  Since $a$ is on the boundary
  of the Voronoi cell, it lies on the bisecting curve of two
  sites, say $s$ and $s'$.  We claim that $C_a$ and $C_b$ lie in
  different sides of $b(s,s')$, which contradicts that $C_a$ and $C_b$
  are contained in the same Voronoi cell of $\kvd$.  For any point $x$
  in $\pi(a,s)$, we have $d(x,s)\leq d(x,s')$. Similarly, for any
  point $x$ in $\pi(a,s')$, we have $d(x,s)\geq d(x,s')$. Also,
  consider the extension of the edge of $\pi(a,s)$ incident to $a$
  towards $a$ until it reaches the boundary of $P$. For any point $x$
  in this extension, we have $d(x,s)\geq d(x,s')$.
  Similarly, we have $d(x,s)\leq d(x,s')$ for any point $x$
  on the extension of the edge of $\pi(a,s')$ incident to $a$
  towards $a$ until it reaches the boundary of $P$.
  This implies that for any point $x$ in $\pi(a,b)$,
  we have $d(x,s)\geq d(x,s')$.  Therefore the claim holds.
	
  For the second part of the lemma, observe that the common boundary
  of the Voronoi cells of any two distinct $k$-tuples $S_k$ and $S_k'$
  of $S$ is a part of the bisecting curve of two sites $s$ and $s'$
  with $s\in S_k\setminus S_k'$ and $s'\in S_k'\setminus S_k$.  Let
  $p$ be an endpoint of a connected component $\Gamma$ of the common
  boundary.  We show that $b(s,s')\setminus \Gamma$ is not contained
  in the common boundary.  Note that $p$ is a degree-3 vertex of
  $\kvd$, which is equidistant from $s, s'$ and another site, say
  $s''$.  Consider the nearest-point Voronoi diagram of $s, s'$ and
  $s''$. Since the common boundary of any two Voronoi cells is
  connected for this diagram~\cite{Aronov-VD-1989}, we have
  $d(s'',x)<d(s,x)=d(s',x)$ for any point $x$ lying on the connected
  component of $b(s,s')\setminus \Gamma$ with endpoint
  $p$. Therefore, the connected component of
  $b(s,s')\setminus \Gamma$ with endpoint $p$ is not
  contained in the common boundary. Similarly, we can prove that the
  other component of $b(s,s')\setminus \Gamma$ is not
  contained in the common boundary, and thus the lemma holds.
\end{proof}

The topological structure of $\kvd$ is defined as follows for
$1\leq k\leq m-1 $.
Imagine that we apply vertex suppression for every degree-2 vertex of
the Voronoi diagram while preserving the topology of the Voronoi
diagram.  Vertex suppression of a vertex $v$ of degree 2 is the
operation of removing $v$ (and the edges incident to $v$) and adding an edge connecting
the two neighboring vertices of $v$.
Then the resulting graph consists of vertices of degree-1 and degree-3
and edges connecting the vertices.
We call the dual of the resulting graph the \emph{adjacency
  graph} of the Voronoi diagram.
The adjacency graph is a planar graph with complexity $O(k(m-k))$, because
the number of degree-1 and degree-3 vertices of $\kvd$ is
$O(k(m-k))$~\cite{LL-KVD-2013}.
The adjacency graph of a Voronoi diagram represents the topological structure
of the Voronoi diagram.

Assume that we have the adjacency graph of the Voronoi diagram
together with the exact positions of the degree-1 and degree-3
vertices of the Voronoi diagram. Each Voronoi cell is defined by $k$ sites,
but any two adjacent Voronoi cells share $k-1$ sites.
Consider two adjacent Voronoi cells $V_1$ and $V_2$.
Let $s_1$ and $s_2$ be the two sites
defining $V_1$ and $V_2$, respectively, which are not shared by them.
Then the common boundary of $V_1$ and $V_2$ is a simple curve connecting
two Voronoi vertices $v$ and $u$ of degree-1 or degree-3 such that
each Voronoi edge in the common boundary is
a part of $b(s_1,s_2)$ lying between $v$ and $u$.
See Figure~\ref{fig:bisect}(b).

To compute the Voronoi edges in the common boundary of $V_1$
  and $V_2$,
we consider the geodesic paths $\pi(s_i',v)$ and $\pi(s_i',u)$ for
$i=1,2$, where $s_i'$ is the junction of $\pi(s_i,v)$ and
$\pi(s_i,u)$.  Then for any vertex $x$ in
$\pi(s_1',v) \cup \pi(s_1',u)$, there exists a point $q$ in $b(s_1,s_2)$ 
lying between $v$ and $u$ such that
$\pi(s_1,q)$ and $\pi(s_1,x)$ have the same combinatorial structure.
The same holds for $s_2$.  This implies that the number of edges in
the common boundary 
is bounded by the total complexity of $\pi(s_1',v) \cup \pi(s_1',u)$
and $\pi(s_2',v) \cup \pi(s_2',u)$. Thus, we compute the geodesic
paths explicitly and consider every edge of the geodesic paths.

Therefore, we can compute 
the common boundary of two adjacent Voronoi cells in time linear to
its complexity plus $O(\log n)$.  This leads to $O(T_1+T_2\log n)$
time for computing the Voronoi diagram from the topological structure
of the diagram, where $T_1$ is the combinatorial complexity of the Voronoi diagram
and $T_2$ is the combinatorial complexity of the adjacency graph.

\begin{lemma}
  \label{lem:topological}
  We can compute the Voronoi diagram of $m$ points in a simple polygon
  with $n$ vertices in $O(T_1+T_2\log n)$ time once the adjacency
  graph and the exact positions of the degree-1 and degree-3
  vertices of the Voronoi diagram are given,
  where $T_1$ is the combinatorial complexity of the
  Voronoi diagram and $T_2$ is the combinatorial complexity of the adjacency graph.
\end{lemma}

Therefore, in the following, we focus on computing the adjacency
  graphs and the exact positions of the degree-1 and degree-3
  vertices 
  of $\vd$, $\kvd$ and $\fvd$.

\section{The Geodesic Nearest-Point Voronoi Diagram}
\label{sec:NVD}
Fortune~\cite{F-VD-1987} presented an $O(m\log m)$-time algorithm to
compute the nearest-point Voronoi diagram of $m$ points in the plane
by sweeping the plane with a horizontal line from top to bottom.
During the sweep, the algorithm computes a part of the Voronoi diagram of sites
lying above the horizontal line, which finally becomes the complete
Voronoi diagram in the end of the sweep. Fortune defined
two types of events and showed how the algorithm processes $O(m)$
events in the order of their $y$-coordinates to compute the Voronoi diagram.
Each event can be handled in $O(\log m)$ time, which leads to $O(m\log m)$
total running time.

In our case, we sweep the polygon with a geodesic path $\pi(o,x)$ for
a fixed point $o$ on the boundary of $P$ and a point $x$ moving along
the boundary of $P$ from $o$ in clockwise order. The point $x$ is
called the \emph{sweep point}.  If we compute all degree-1,
degree-2 and degree-3 vertices of the Voronoi diagram during the
sweep, we may not achieve the running time better than
$O((n+m)\log(n+m))$ as there are $O(n+m)$ such vertices.
The key to improve the running time is to
compute the topological structure of the Voronoi diagram first which
consists of the degree-1 and degree-3 vertices of the Voronoi diagram
and the adjacency of the Voronoi cells.  Then we construct the complete
Voronoi diagram, including degree-2 vertices, from its topological
structure using Lemma~\ref{lem:topological}.

Let $o$ be an arbitrary point on $\bd P$, where $\bd P$ denotes the
boundary of $P$.  Consider the sweep point $x$ that moves from $o$
along $\bd P$ in clockwise order.  We use $P(x)$ to denote the
subpolygon of $P$ bounded by $\pi(o,x)$ and the part of $\bd P$ from
$o$ to $x$ in clockwise order. In other words, $P(x)$ is the region
swept by $\pi(o,x)$. Note that $P(x)$ is weakly simple. See
Figure~\ref{fig:events}. 
Clearly, $P(x_1)\subseteq P(x_2)$ for any two points $x_1$
and $x_2$ on $\bd P$ such that $x_1$ comes before $x_2$ from $o$ in
clockwise order along $\bd P$.

For a site $s \in P(x)$, let $R_s(x)$ be the set
$\{p\in P(x)\mid d(p,s)\leq d(p,\pi(o,x))\}$.  By definition,
$R_s(x_1)\subseteq R_s(x_2)$ for any two points $x_1$
and $x_2$ on $\bd P$ such that $x_1$ comes before $x_2$ from $o$ in
clockwise order along $\bd P$.

\begin{lemma}
  \label{lem:site-on-path}
If $s$ lies on an edge of $\pi(o,x)$, $R_s(x)$ is a line segment
that is incident to $s$ and orthogonal to the edge.
\end{lemma}
\begin{proof}
Let $\gamma_s$ be the boundary of $R_s(x)$.  By
  definition, $d(p,s)=d(p,\pi(o,x))$ for any point $p$ in $\gamma_s$.
  Note that $d(p,s)=d(p,\pi(o,x))$ if and only if $s$ is the 
  point of $\pi(o,x)$ closest to $p$ under the geodesic metric.
  This is because $s$ lies on $\pi(o,x)$.
  Therefore, for any point $p$ in $\gamma_s$, the 
  edge of $\pi(s,p)$ incident to $s$ is orthogonal to the edge of $\pi(o,x)$
  containing $s$.
	
Moreover, we claim that $\gamma_s$ consists of a single line
  segment.  If $\pi(p,s)$ consists of more than one line segment for
  some point $p$ on $\gamma_s$, we can choose a sufficiently small
  neighborhood $N_p$ of $p$ contained in $P'$ 
  such that the point on $\pi(o,x)$ closest to $p'$ 
  is $s$ for any point $p' \in N_p$.  This implies that
  $N_p$ is contained in $R(x)$, which contradicts that $\gamma_s$ is a part
  of the boundary of $R(x)$.
\end{proof}

\begin{lemma}
  \label{lem:R-connect}
  For a site $s\in P(x)$, $R_s(x)$ is connected. 
\end{lemma}
\begin{proof}
  If $s\in\pi(o,x)$, $R_s(x)$ is a line segment by Lemma~\ref{lem:site-on-path},
  and therefore it is connected.
  In the following, we assume that $s$ is not on $\pi(o,x)$ and 
  show that $\pi(p,s)\subseteq R_s(x)$ for any
  point $p\in R_s(x)$.  This implies that $R_s(x)$ is connected
  because $s$ is contained in $R_s(x)$ if $s\in P(x)$ by definition.
  Let $p$ be a point in $R_s(x)$.  Consider a point $r \in \pi(p,s)$.
  We have $d(r,s)=d(p,s)-d(p,r)$.  Moreover, we have
  $d(p,\pi(o,x)) - d(p,r)\leq d(r,\pi(o,x))$.  Since $p \in R_s(x)$,
  it holds that $d(p,s)\leq d(p,\pi(o,x))$.  Thus,
  $d(r,s) \leq d(r,\pi(o,x))$, and $r$ is also in $R_s(x)$ by definition.
  Therefore $R_s(x)$ is connected.
\end{proof}

We say that a subset $A$ of $P$ is \emph{weakly monotone}
with respect to a geodesic path $\gamma$ if the intersection of
$\pi(p, \gamma)$ with $A$ is connected for any point $p \in A$.
We define a \emph{shaft} 
from a point $x\in P$ towards a direction
to be the line segment connecting $x$ and $y$,
where $y$ is the first intersection point of the boundary of $P$
with the ray from $x$ towards the direction.

\begin{lemma}
  \label{lem:R-monotone}
  For a site $s\in P(x)\setminus \pi(o,x)$, the boundary of $R_s(x)$
  consists of one polygonal chain of $\bd P(x)$ and a simple curve whose
    both endpoints lie on $\bd P(x)$.
  Moreover, the simple curve is weakly monotone with
  respect to $\pi(o,x)$.
\end{lemma}
\begin{proof} 
  By Lemma~\ref{lem:R-connect}, $R_s(x)$ is connected. Thus to prove the first part
  of the lemma, it suffices to show that the boundary of $R_s(x)$ intersects $\bd P(x)$. 
  To do this, consider the shaft from a point $p\in R_s(x)$ in direction opposite to the
  edge of $\pi(p,\pi(o,x))$ incident to $p$. 
  We claim that the shaft is contained in $R_s(x)$.
   This is because
  $d(r,\pi(o,x))=d(p,\pi(o,x))+d(p,r) \geq d(p,s)+d(p,r)\geq d(r,s)$
  for any point $r$ in the shaft   by triangle inequality and the fact that
  $R_s(x)$ contains no point of $\pi(o,x)$. Note that an endpoint of the shaft is on $\bd P(x)$.
  Therefore, $R_s(x)$ intersects $\bd P(x)$, and the first part of the lemma holds.
  
  For the second part of the lemma, we claim that $\pi(p,\pi(o,x))\setminus \{p\}$
  does not intersect $R_s(x)$ for any point $p$ on $\bd R_s(x)\setminus \bd P(x)$. 
  Assume to the contrary that
  a point $p'$ in the path is in $R_s(x)$. We already showed that the
  shaft from $p'$ in direction opposite to the
  edge of $\pi(p',\pi(o,x))$ incident to $p'$ is contained in $R_s(x)$. Note that the shaft
  contains $p$. This contradicts that $p$ lies on the boundary of $R_s(x)$.
  Therefore, the second part of the lemma also holds.
\end{proof}

Now we consider the union of $R_s(x)$ for every site $s$ in $P(x)$
and denote it by $R(x)$. Figure~\ref{fig:events} shows three sites contained
  in $P(x)$ for a simple polygon $P$. The dashed region in Figure~\ref{fig:events}(a) is $R(x)$. 
Note that for any point $p \in P(x)$ and any site $s$ lying outside of $P(x)$,
it holds that $d(p,\pi(o,x))< d(p,s)$.
This implies that for any point $p\in R(x)$, the nearest site of $p$
is in $P(x)$. Therefore, for computing $\vd$ restricted to $R(x)$, we do not need
to consider any sites lying outside of $P(x)$. 

The following corollaries follow from the properties of $R_s(x)$.
\begin{corollary}
  \label{lem:beach-line-weakly}
  The closure of $P(x)\setminus R(x)$ is weakly simple.
\end{corollary}

\begin{corollary}
  \label{lem:beach-line-order}
  Each connected component of $R(x)$ consists of a polygonal chain of
  $\bd P$ and a simple curve with endpoints on $\bd P$ unless
    $\pi(o,x)$ contains a site.
  Moreover, the
  union of such curves is weakly monotone with respect to $\pi(o,x)$
  at any time.
\end{corollary}

We maintain the geodesic nearest-point Voronoi diagram restricted to $R(x)$
of the sites contained in $P(x)$ while $x$ moves along $\bd P$.
However, after the sweep, $R(x)$ is still a proper subpolygon $P$, and therefore
$\vd$ is not completed yet. To resolve this, we attach a long and very thin
triangle to $P$ to make a bit larger simple polygon $P'$ and set $o$ to lie at the tip
of the triangle, as illustrated in Figure~\ref{fig:events},
so that once we finish the sweep (when the sweep point returns back to $o$) in $P'$
we have the complete Voronoi diagram in $P$.
(The triangle has height of the diameter of $P$.)

\paragraph{The beach line and the breakpoints.}  
There are $O(m)$ connected components of $R(x)$ each of whose boundary
consists of a polygonal chain of $\bd P$ and a connected simple curve.
The \emph{beach line} is defined to be 
the union of the $O(m)$ simple curves of $R(x)$.
The beach line has properties similar to those of the beach line of the Euclidean Voronoi
diagram.  It consists of $O(n+m)$ hyperbolic or
linear arcs.  We do not maintain them explicitly because the complexity of
the sequence is too large for our purpose.  Instead,
we maintain the combinatorial structure of the beach
line using $O(m)$ space as follows.

A point on the beach line is called a \emph{breakpoint} 
if it is equidistant from two distinct sites. 
Each breakpoint
moves and traces out a Voronoi edge as the
sweep point moves along $\bd P'$.  
We represent each breakpoint symbolically.
That is, a breakpoint is represented as the pair of sites
equidistant from it.
We say that such a pair \emph{defines} the breakpoint.
Given the pair defining a breakpoint and
the position of the sweep point,
we can find the exact position of the breakpoint
using Lemma~\ref{lem:equidistant-line} in $O(\log^2 n)$ time.

Additionally, we consider the endpoints of the $O(m)$ simple curves
comprising the beach line. We simply call them the \emph{endpoints} of the beach line.
For an endpoint $p$ of the beach line, there is a unique site $s$
satisfying $d(p,s)=d(p,\pi(o,x))$. We say that $s$ \emph{defines} $p$.

By Corollary~\ref{lem:beach-line-weakly} and
Corollary~\ref{lem:beach-line-order}, we can define
the order of these breakpoints and these endpoints. Let $B(x)=\langle
\beta_1,\ldots,\beta_{m'}\rangle$ be the sequence
of the breakpoints and the endpoints of the beach line
sorted in clockwise order along the boundary
of $P(x)\setminus R(x)$ with $m'=O(m)$.
We maintain $B(x)$ instead of the arcs on the beach line.
As $x$ moves along $\bd P'$, the
sequence $B(x)$ changes.  We will see that the
length of $B(x)$ is $O(m)$ at any time during the
move of the sweep point along $\bd P'$.

While maintaining $B(x)$, 
we compute the adjacency graph of $\vd$ together with the degree-1 and 
degree-3 vertices of $\vd$ restricted to $R(x)$.
Specifically, when a new breakpoint defined by 
a pair $(s,s')$ of sites is 
added to $B(x)$, we add an edge connecting the node for
$s$ and the node for $s'$ into the adjacency graph.

\paragraph{Events.}  
\begin{figure}
  \begin{center}
    \includegraphics[width=\textwidth]{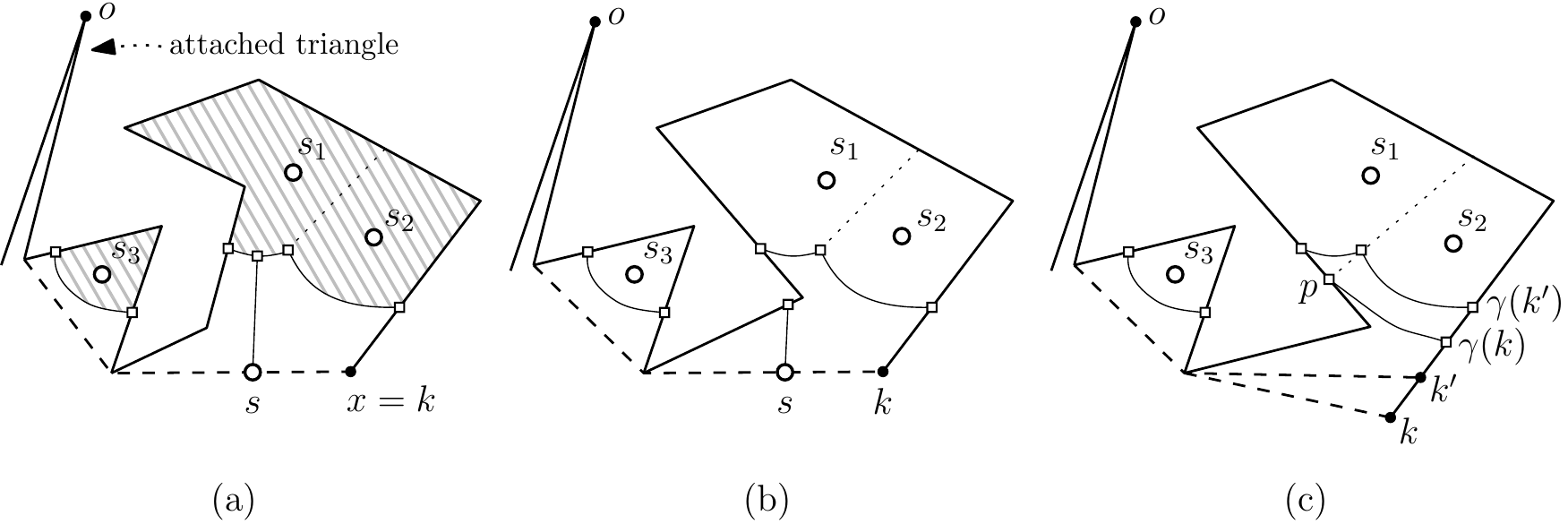}
    \caption {
      (a) The site event defined by $s$ with key $k$.
      Two (degenerate) breakpoints appear in $B(k)$.
      (b) Two (degenerate) endpoints appear in $B(k)$.
      (c) The vanishing event defined by the breakpoint of $(s_1,s_2)$
        with key $k$.
        The endpoint defined by $s_1$ and the breakpoint merge
        into the endpoint defined by $s_2$.
      \label{fig:events}}
  \end{center}
\end{figure}
We have four types of
events: site events, circle events, vanishing events, and 
merging events.
Every event corresponds to a \emph{key}, which is a
point on the boundary of $P'$.
We maintain the events with respect to their 
keys sorted in clockwise order from $o$ along the boundary.
The event occurs when
the sweep point $x$ passes through the 
corresponding key.  The sequence $B(x)$ changes only
when $x$ passes through the key of an event.
Given a sorted sequence of $\nu$ events for $\nu\in\mathbb{N}$,
we can insert a new event in $O(\log \nu)$ time since
we are given each point on the boundary of $P'$ together with the edge of $P'$ 
where it lies.
We can delete an event from $\nu$ in $O(\log |\nu|)$ time, and find 
the first event in $\nu$ and delete it in $O(1)$ time.

The definitions of the first two event types are
similar to the ones in Fortune's algorithm~\cite{F-VD-1987}.
Each site $s$ in $S$ defines a \emph{site event}.  The key $k$ 
of the site event defined by $s$ is the point on $\bd P'$
that comes first from $o$ in clockwise order along $\bd P'$ among the points $x'$ 
with $s\in \pi(o,x')$. When the sweep point passes through $k$,
the site $s$ appears on the beach line. At the same time, new breakpoints defined by
$s$ and some other site $s'$, or new endpoints defined by $s$
appear on $B(x)$.
See Figure~\ref{fig:events}(a) and (b).

An event of the other three event types is 
defined by a pair of consecutive points in $B(x)$
or a single breakpoint in $B(x)$.
An event is said to be \emph{valid} if 
the two points defining the event are consecutive,
or the breakpoint defining the event is in $B(\cdot)$.
Before the sweep point reaches the key of an event,
the two points defining the event may become non-consecutive,
or the breakpoint defining the event may disappear from $B(x)$
due to the changes of $B(\cdot)$.
In this case, we say that the event become \emph{invalid}.

A pair $(\beta_1,\beta_2)$ of consecutive breakpoints in
$B(x)$ defines a \emph{circle event}
if there is a point $c$ equidistant from $s_1,s_2$ and $s_3$ under the geodesic metric,
where $(s_1,s_2)$ and $(s_2,s_3)$ are two pairs of sites
defining $\beta_1$ and $\beta_2$, respectively.
The key $k$ of this circle event is the point on $\bd P'$
that comes first from $o$ in clockwise order along $\bd P'$
among the points $x'$ with $d(c,\pi(o,x'))=d(c,s_1)(=d(c,s_2)=d(c,s_3))$.  
If this event is valid when $x$ passes through $k$,
$c$ appears on the beach line at the time.
Moreover, $\beta_1$ and $\beta_2$
disappear from the beach line, and a new breakpoint defined
by $(s_1,s_3)$ appears on the beach line.
(One may regard this as merging $\beta_1$ and $\beta_2$ into the breakpoint
defined by $(s_1,s_3)$.)

Each breakpoint $\beta$ in $B(x)$ defines a \emph{vanishing event}. 
Let $(s_1,s_2)$ be the pair of sites defining $\beta$.
Consider the two points on $\bd P'$ equidistant from $s_1$ and $s_2$.
We observe that exactly one of them lies outside of $R(x)$.
We denote it by $p$. See Figure~\ref{fig:events}(c).
The key $k$ of the vanishing event is the point on $\bd P'$  
that comes first from $o$ in clockwise order along $\bd P'$
among the points $x'$ with $d(p,\pi(o,x'))=d(p,s_1)$.
Assume that this event is valid when $x$ passes through $k$.   
Then, $\beta$ traces out a Voronoi edge and reaches $p$, which is
a degree-1 Voronoi vertex.
Moreover, $B(x)$ changes accordingly as follows.
Right before $x$ reaches $k$, 
an endpoint defined by $s_1$ or $s_2$ is a neighbor of $\beta$ 
in $B(x)$. (Otherwise, the beach line is not weakly monotone.)
Without loss of generality, we assume that an endpoint defined by $s_1$
is a neighbor of $\beta$.
When $x$ reaches the key, this endpoint and $\beta$ disappear from $B(x)$, and
an endpoint defined by $s_2$ appears in $B(x)$.
(One may regard this as merging the endpoint defined by $s_1$ and $\beta$ into the endpoint
defined by $s_2$.)

A pair $(\beta_1, \beta_2)$ of 
consecutive endpoints in $B(x)$ defines a \emph{merging event} if
$\beta_1$ and $\beta_2$ are endpoints of two distinct connected curves
of the beach line.
Let $s_1$ and $s_2$ be the sites defining $\beta_1$ and $\beta_2$, respectively.
Without loss of generality, 
assume that $\beta_1$ comes before $\beta_2$ from $o$ in clockwise order
along $\bd P'$.
Let $p$ be the (unique) point on $\bd P'$ equidistant from $s_1, s_2$ that comes after $\beta_1$
and before $\beta_2$ from $o$  in clockwise order along $\bd P'$.
The key $k$ of the merging event 
is the point on $\bd P'$
that comes first from $o$ in clockwise order along $\bd P'$
among the points $x'$ with $d(p,\pi(o,x'))=d(p,s_1)$.
Assume that this event is valid when $x$ passes through $k$.
Then, as the sweep point $x$ moves along $\bd P'$, 
$\beta_1$ and $\beta_2$ become closer and finally meet each other at $p$.
This means that the two connected curves, one containing $\beta_1$ and
the other containing $\beta_2$, merge into one connected curve.
At this time, $\beta_1$ and $\beta_2$ merge into a breakpoint
defined by $(s_1,s_2)$.

\subsection{An Algorithm}
Initially, $B(x)=B(o)$ is empty, so there is no circle, vanishing, or merging event.  We compute the keys of all site events in advance. For each site, we
compute the key corresponding to its event in $O(\log n)$
time~\cite{shortest-path-tree}.
As $B(x)$ changes, we obtain new pairs of consecutive breakpoints or endpoints,
and new breakpoints.  Then we
compute the key of each event in $O(\log^2 n)$ time using the following
lemmas and Lemma~\ref{lem:equidistant}.
In addition, as $B(x)$ changes, some event becomes invalid.
Once an event becomes invalid, it does not become valid again.
Thus we discard an event if it becomes invalid.
Therefore, the number of events we have is $O(|B(x)|)$ at any time.

\begin{lemma}
  \label{lem:circle-key}
  Given a point $p$ in $P'$ and a distance $r \in \mathbb{R}$, the
  point that comes first from $o$ in clockwise order
    along $\bd P'$
  among the points $x'\in\bd P'$ with $d(p,\pi(o,x'))=r$
  can be found in $O(\log^2 n)$ time.
\end{lemma}
\begin{proof}
  We compute a point $x_p$ on $\bd P'$ such that $\pi(o,x_p)$ contains
  $p$ in $O(\log n)$ time~\cite{shortest-path-tree}.  Then
  $d(p,\pi(o,x_p))=0$ by definition.  As the sweep point $x$ moves
  along $\bd P'$ from $x_p$ in clockwise order, $d(p,\pi(o,x))$ does not
  decrease. This is because the point $p'\in \pi(p,\pi(o,x_2))$ closest to $p$
  does not lie in the interior of $P(x_1)$ for
  two points $x_1$ and $x_2$ such that $x_1$ comes before
  $x_2$ from $o$ in clockwise order along $\bd P'$.
  Let $x^*$ be the point that comes first from $o$ in clockwise order
  along $\bd P'$ among the points $x'$ with $d(p,\pi(o,x'))=r$.
  By the monotonicity of $d(p,\pi(o,\cdot))$, 
  we can apply
  binary search on the vertices of $\bd P'$ to compute the edge of $P'$
  containing $x^*$.  Since we can compute the geodesic distance from
  $p$ to $\pi(o,x')$ for any point $x'$ in $\bd P'$ in $O(\log n)$
  time, we can find the edge $e$ of $P'$ containing $x^*$ in $O(\log^2
  n)$ time by applying binary search on the vertices of $\bd P'$.
	
  Consider the subdivision of $e$ by the extensions of the edges of
  the geodesic paths between $o$ and each endpoint of $e$. The subdivision
  consists of $O(n)$ disjoint line segments on $e$. We apply binary search
  to find the line segment on $e$ containing $x^*$ in $O(\log^2 n)$
  time without computing the subdivision explicitly as we did in Section~\ref{sec:center-three}.
  For any point $x'$ in the line segment, $\pi(o,x')$ has the same combinatorial structure.  
  Thus we can compute $x^*$ directly in constant time.
  
  Therefore, we can compute the point on $\bd P'$
  that comes first from $o$ in clockwise order 
  along $\bd P'$
  among the points $x'$ with $d(p,\pi(o,x'))=r$ in $O(\log^2 n)$ time in total.
\end{proof}

\begin{lemma}
  \label{lem:equidistant-boundary}
  For any two sites $s_1$ and $s_2$ in $P'$, we can compute the two
  points equidistant from $s_1$ and $s_2$ under the geodesic metric
  that lie on the boundary of $P'$ in $O(\log^2 n)$ time.
\end{lemma}
\begin{proof}
  We first compute the balanced binary search tree of $\pi(s_1,s_2)$ in $O(\log n)$ time. 
   Then we extend
  the edge of $\pi(s_1,s_2)$ incident to $s_i$ towards
  $s_i$ for $i=1,2$ until it reaches a point $x_i$ on $\bd P'$.
  Then $x_1$ and $x_2$ partition $\bd P'$ into two parts,
  one from $x_1$ to $x_2$ and one from $x_2$ to $x_1$ in clockwise
  order. Note that each part contains one point equidistant from $s_1$ and $s_2$
  under the geodesic metric.
  We show how to compute the equidistant point $p$
  on the part from $x_1$ to $x_2$ in clockwise order in $O(\log^2 n)$ time.
	
  To compute the edge of $P'$ containing $p$, we apply binary search on
  the vertices of $\bd P'$ lying from $x_1$ to $x_2$ in clockwise
  order using the property that 
  $d(x,s_1)\leq d(x,s_2)$ for any point $x$ on $\bd P'$ from $x_1$ to $p$
  in clockwise order, and $d(x, s_1)\geq d(x,s_2)$ for any point
  $x$ on $\bd P'$ from $p$ to $x_2$ in clockwise order. 
  We find the edge $e$ of $P'$ containing $p$ in $O(\log^2 n)$ time.

  Then we consider the subdivision of the edge containing $p$ by the extensions of the
  edges of $\pi(s_1,b)\cup\pi(s_1,b')\cup\pi(s_2,b)\cup\pi(s_2,b')$
  for two endpoints $b$ and $b'$ of the edge $e$.
  The subdivision consists of $O(n)$ disjoint line segments on $e$. 
  By applying binary search without computing the subdivision explicitly,
  we find the line segment containing $p$ in $O(\log^2 n)$ time.  Inside the line segment, we
  find $p$ directly in constant time.
  Therefore, we can find $p$ and $p'$ in total $O(\log^2 n)$ time.
\end{proof}

\paragraph{Handling site events.}  
To handle a site event defined by a site $s$ with key $k$, we do the following.
By definition, $s$ appears on the beach line
when the sweep point $x$ passes through $k$.
By Lemma~\ref{lem:site-on-path}, the part of the beach line defined by
  a site $s$ is a single line segment when the sweep point passes through $k$.
  Therefore, $B(k)$ contains two (degenerate) breakpoints
  defined by $s$ and some other sites or two (degenerate) endpoints
  defined by $s$, which lie at the endpoint 
  of the single line segment other than $s$. See Figure~\ref{fig:events}(a) and (b).
We find the positions of them 
in $B(x)$ and update $B(x)$ by adding them using Lemma~\ref{lem:com-break-site}.

\begin{lemma}
  \label{lem:com-break-site}
  We can obtain $B(k)$ from
  $B(k')$ in $O(\log^2 n\log |B(k')|)$ time,  where $k'$ is the key previous to $k$.
\end{lemma}
\begin{proof}
  By applying binary search on the breakpoints and the endpoints on $B(k')$,
  we compute the part (a single line segment) $\gamma_s$ of the beach line
  defined by $s$ when the sweep point $x$ passes through $k$ as follows. 
  Let $\beta$ be a breakpoint or an endpoint of the beach line on $B(k')$.
  We determine the position of $\gamma_s$ on the beach line $\gamma$ 
  with respect to $\beta$ in $O(\log^2n)$ time.
  Recall that we maintain $\beta$ symbolically. We have two
  sites (or one site if $\beta$ is an endpoint) defining $\beta$,
  but not the exact position of $\beta$.
  If $\beta$ is a breakpoint, 
  we compute the exact position of $\beta$ using Lemma~\ref{lem:equidistant-line}.
  For the case that $\beta$ is an endpoint, 
  we can compute the exact position of $\beta$
  in $O(\log^2 n)$ time in a way similar to the one
  in Lemma~\ref{lem:equidistant-line}. 
  The order of $\gamma_s$ and $\beta$ along $\gamma$
  is the same as the order
  of $s$ and the point on $\pi(o,k)$ closest to $\beta$ under the geodesic metric
  by Lemma~\ref{lem:beach-line-order}.
  Thus we can compute the order of $\gamma_s$ and $\beta$ 
  along $\gamma$ in $O(\log^2 n)$ time in total.
  Using this property, we apply binary search on the breakpoints
  and the endpoints on $B(k')$,
  and find the position for $\gamma_s$ on the beach line in
  $O(\log^2 n \log |B(k')|)$ time.
  
  Based on this result, we compute $B(k)$ by adding two (degenerate) breakpoints
  or two (degenerate) endpoints to $B(k')$.
  In any case, we can compute $B(k)$ in $O(\log |B(k')|)$ time.  
  Thus, the total running time is $O(\log^2 n\log|B(k')|)$.
\end{proof}

Adding two breakpoints or two endpoints to $B(x)$ 
makes a constant number of new pairs of consecutive breakpoints or endpoints, which
define circle or merging events.
This also makes a constant number of new vanishing events.
We compute the key for each such
event in $O(\log^2 n)$ time and add the events to the event sequence
in $O(\log |B(x)|)$ time.
Recall that the number of events we have is $O(|B(x)|)$ at any time.
Therefore, each site event can be handled in $O(\log^2 n\log |B(x)|)$ time.

\paragraph{Handling the other events.}  Let $k$ be the key
of a circle, vanishing, or merging event. 
For a circle event,
two breakpoints defining the event disappear from the beach line,
and a new breakpoint appears. 
For a vanishing event, 
the breakpoint defining the event and a neighboring endpoint 
disappear from the beach line, and a new endpoint appears. 
For a merging event, two endpoints defining the event
are replaced with a new breakpoint. 
In any case, we can update $B(x)$ in $O(\log |B(x)|)$ time.

After updating $B(x)$, we have a constant number of new pairs of
consecutive breakpoints or endpoints, and 
a constant number of new breakpoints.
 We compute 
the keys of events defined by them in $O(\log^2 n)$ time.

\paragraph{Analysis.} For analysis of the correctness,
we show that the combinatorial structure of the beach line changes
only when the sweep point $x$ passes through an event key.

\begin{lemma}
The combinatorial structure of the beach line changes
only when the sweep point $x$ passes through an event key.
\end{lemma}
\begin{proof} 
  Consider a breakpoint $\beta$ on $B(x)$
  defined by two sites $s_1$ and $s_2$ with the sweep point 
  at a point $x$ on $\bd P'$.
  If $s_1$ or $s_2$ is on $\pi(o,x)$, $x$ is the key of a site event
  defined by $s_1$ or $s_2$. In the following,
  we assume that $s_1, s_2\in P(x)$ but not on $\pi(o,x)$.
  Then $\beta$ lies on the common boundary of the Voronoi cells of
  $s_1$ and $s_2$ or a degree-3 Voronoi vertex of $\vd$ defined by three
  sites including $s_1$ and $s_2$.
  
  Consider the case that $\beta$ lies on the common boundary of the Voronoi
  cells other than its endpoints. 
  We claim that
  the combinatorial structure does not change in this case.
  There are two cases: $\beta$ lies on a Voronoi edge or a degree-2 Voronoi vertex.
  For the first case, we let $e$ be the edge containing $\beta$. For the second case,
  we let $e$ be the Voronoi edge incident to the degree-2 Voronoi vertex such that
  $d(\beta,s_1)<d(x,s_1)$ for any point $x\in e$.
  Let $U_p$ denote the union of the shafts from $p$ in direction opposite to
  the edges of $\pi(p,\pi(o,x))$ incident to $p$ for every point
  $p\in \pi(s_1,\beta)\cup\pi(s_2,\beta)$. Then  $U_p\subseteq
  R_{s_1}(x)\cup R_{s_2}(x)$ and $U_p$ intersects $e$
  at a point other than $\beta$ unless
  $s_1\in \pi(\beta,\pi(o,x))$ or $s_2\in \pi(\beta,\pi(o,x))$. 
  (But, it is not possible that $s_1,s_2\in \pi(\beta,\pi(o,x))$ since
  $d(\beta,s_1)=d(\beta,s_2)=d(\beta,\pi(o,x))$ and $x$ is not
  the key of a site event.)
  By continuity of $R_s(\cdot)$ for any site $s\in P(\cdot)$,
  this implies that there is a breakpoint defined by $s_1$ and $s_2$
  before the sweep point reaches $x$.
  Thus, the combinatorial structure does not change in this case.
	
  Consider the case that $\beta$ lies on a degree-3 Voronoi vertex.
  We claim that $x$ is the key of a circle event.
  Let $s_3$ be the site such that $\beta$ is defined by $s_1,s_2$ and $s_3$.
  We again consider the union $U_p$ of the shafts from $p$
  in direction opposite to the edges of $\pi(p,\pi(o,x))$ incident to $p$
  for every point $p\in \pi(s_1,\beta)\cup\pi(s_2,\beta)\cup\pi(s_3,\beta)$. 
  Then $U_p\subseteq R_{s_1}(x)\cup R_{s_2}(x)\cup R_{s_3}(x)$ and
  $U_p$ intersects two Voronoi edges incident to $\beta$.
  This implies that there are two consecutive breakpoints defined by $(s_1,s_2)$
  and $(s_2,s_3)$ before the sweep point reaches $x$.
  So, $x$ is the key of the circle event defined by this pair of breakpoints.
  
  We can show that the same holds for the case of an endpoint $\beta$
  in a similar way. If $\beta$ lies on $\bd P'$
  which is not a Voronoi vertex, the combinatorial structure does not change.
  If $\beta$ lies on a degree-1 Voronoi vertex,
  $x$ is the key of a vanishing event or merging event. 
\end{proof}

For analysis of the running time, we give
upper bounds on the length of $B(x)$ and the total number of 
valid events we have handled.

\begin{lemma}
  \label{lem:length}
  The length of $B(x)$ is $O(m)$ at any time.
\end{lemma}
\begin{proof}
  Initially, the length of $B(x)$ is zero. The length of
  $B(x)$ increases if the sweep point $x$ passes through the key of a site event.
  In this case, we create at most two breakpoints or two endpoints. 
  Thus, the length of $B(x)$ increases by at most two.  
  For other events, the length of $B(x)$ decreases.  
  Since we have $O(m)$ site events, the length of
  $B(x)$ is $O(m)$ at any time.
\end{proof}

An event becomes invalid in the process of handling another event.
The time we spend to discard an event from $B(x)$ is subsumed by the time
we spend to handle the event that invalidates it.
Thus, it is sufficient to bound the number of events
that are valid at the time the sweep point $x$ passes through their
corresponding keys. 

\begin{lemma}
  There are $O(m)$ events in total that are valid at the times the sweep point $x$
  passes through their corresponding keys.
\end{lemma}
\begin{proof}
  Consider the sites that are valid at the times the sweep point $x$ passes
  through their corresponding keys.
  Each site event corresponds to a site in $S$,
  and each site defines exactly one site event.
  Each circle event  corresponds to a degree-3 vertex of $\vd$.
  Each vanishing or merging event corresponds to a degree-1 vertex of $\vd$.	
  Since we have $m$ sites and $O(m)$ degree-1 or degree-3 vertices,
  the total number of events is $O(m)$.
\end{proof}

Therefore, We have the following lemma and theorem.
\begin{lemma}
  Once the shortest path data structure for $P'$ and the shortest path
  map for $P'$ from a fixed point are constructed,
  we can compute the topological structure
  of $\vd$ in $O(m\log m\log^2 n)$ time using $O(m)$ space.
\end{lemma}

\begin{theorem}
  Given a set of $m$ point sites contained in a simple polygon with
  $n$ vertices, we can compute the geodesic nearest-point Voronoi
  diagram of the sites in $O(n+m\log m\log^2n)$ time using $O(n+m)$
  space.
\end{theorem}

\section{The Geodesic Higher-order Voronoi Diagram}
In this section, we first present an asymptotically tight combinatorial 
complexity of the geodesic higher-order Voronoi diagram of points in a simple polygon.
Then we present an algorithm to compute the diagram by applying the
polygon-sweep paradigm introduced in Section~\ref{sec:NVD}.  In the
plane, Zavershynskyi and Papadopoulou~\cite{higherVD3} presented a
plane-sweep algorithm to compute the higher-order Voronoi diagram. We
use their approach together with our approach in
Section~\ref{sec:NVD}.  In the following, we assume that $1\leq k\leq m-1$.

\subsection{The Complexity of the Diagram inside a Simple Polygon}
Liu and Lee~\cite{LL-KVD-2013} presented an asymptotically tight 
complexity of the geodesic higher-order Voronoi diagram of
points in a polygonal domain with holes.  However, 
their bound is not tight for a simple polygon.
We present an improved upper bound and prove that it is asymptotically tight.

\paragraph{A lower bound.}
\begin{figure}
  \begin{center}
    \includegraphics[width=0.9\textwidth]{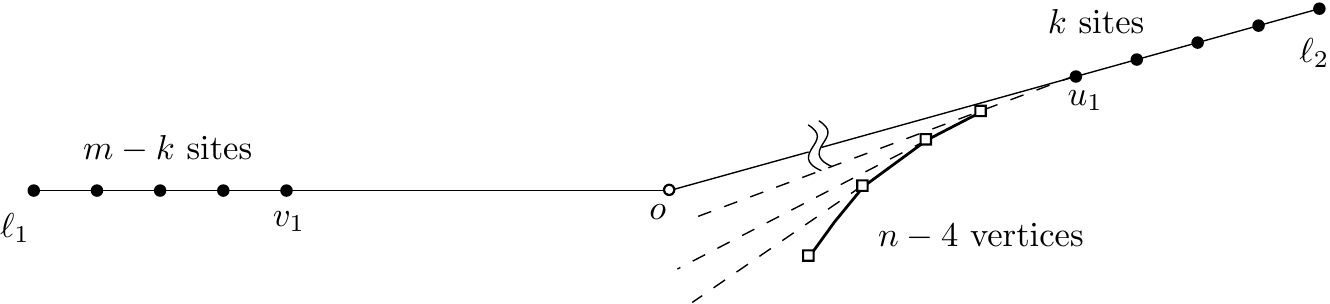}
	\caption {
	An example of which the order-$k$ Voronoi diagram has
	complexity of $\Omega(k(m-k)+\min\{nk, n(m-k)\})$.
	\label{fig:lowerbound}}
    \end{center}
\end{figure}
Figure~\ref{fig:lowerbound} shows an example of which
the order-$k$ Voronoi diagram has complexity of
$\Omega(k(m-k)+\min\{nk, n(m-k)\})$.
Let $o$ be an arbitrary point in the plane. We
construct a simple polygon $P$ and a set of sites with
respect to the point $o$.  Let $\ell_1$ be a
sufficiently long horizontal line segment whose
right endpoint is $o$ and $\ell_2$ be a
sufficiently long line segment with a positive slope 
 whose left endpoint is $o$.  
($\ell_1$ and $\ell_2$ are not parts of the simple polygon,
but they are auxiliary line segments to locate the points.)
We put $m-k$ sites on $\ell_1$ such that the sites are sufficiently
close to each other. 
Similarly, we put $k$ sites on $\ell_2$ such that sites are sufficiently
close to each other and $\|v_{m-k}-o\|\ll\|u_1-o\|$, where
$v_i$'s are sites on $\ell_1$ sorted along $\ell_1$
from $o$, $u_j$'s are sites on $\ell_2$ sorted
along $\ell_2$ from $o$ for $i=1,\ldots,m-k$ and $j=1,\ldots,k$.

As a part of the boundary of $P$, 
we put a concave and $y$-monotone polygonal curve with
$n-4$ vertices below $\ell_2$ such that the highest point of the curve 
is sufficiently close to $u_1$.  Then we put another
four vertices of $P$ sufficiently far from the sites such that $P$ contains all  sites and there is no simple polygon 
containing more Voronoi vertices than $P$ contains.

\begin{lemma}
    \label{lem:lower-cells}
    For any $\kappa$ consecutive sites on $\ell_1$ with $1\leq \kappa\leq \min\{k,m-k\}$,
    there exist $k-\kappa$ sites on $\ell_2$ such that
    the $\kappa$ sites on $\ell_1$ and the $k-\kappa$ 
    sites on $\ell_2$ define a non-empty Voronoi cell.
\end{lemma}
\begin{proof}
    Consider $\kappa$ consecutive sites
    $v_i,\ldots,v_{i+\kappa-1}$ on $\ell_1$.  We show that there
    is a point whose nearest $k$ sites are
    $v_i,\ldots,v_{i+\kappa-1}$ and
    $u_1,\ldots,u_{k-\kappa}$.  To show this, consider the point $x$
    equidistant from $v_i, v_{i+\kappa-1}$ and
    $u_{k-\kappa}$.  Since $\|v_{m-k}-o\|\ll\|u_1-o\|$, it holds that $\|x-u_{k-\kappa}\|\geq
    \|x-u_j\|$ for any index $j \leq k-\kappa$
    and $\|x-u_{k-\kappa}\|\leq
    \|x-u_j\|$ for any index $j \geq k-\kappa$. This
    implies that the nearest $k$ sites from $x$ is
    $v_i,\ldots,v_{i+\kappa-1}$ and
    $u_1,\ldots,u_{k-\kappa}$.  Therefore, the lemma
    holds.
\end{proof}

\begin{lemma}
    \label{lem:lower-edges}
    For every integer $\kappa$ with $1\leq \kappa \leq \min\{k, m-k\}$, there are $\Theta(n)$ Voronoi
    edges which come from the bisecting curve of $v_\kappa$
    and $u_{k-\kappa}$.
\end{lemma}
\begin{proof}
    Consider the bisecting curve of $v_\kappa$ and
    $u_{k-\kappa}$.  Due to the concave curve with
    $n-4$ vertices, the bisecting curve consists of $O(n)$
    hyperbolic arcs and one line segment.  We claim
    that each hyperbolic arc is a Voronoi edge of
    $\kvd$.  To see this, for any point $p$ in a
    hyperbolic arc of the bisecting curve, observe that
    $d(v_\kappa,p) \leq d(v_i,p)$ for any $i \geq
    \kappa$ and $d(u_{k-\kappa},p) \leq d(u_j,p)$ for
    any $j\geq k-\kappa$.  
    This is because $p$ lies below $\ell_1\cup\ell_2$.
    Therefore, there are $O(n)$
    Voronoi edges from the bisecting curve of $v_\kappa$ and
    $u_{k-\kappa}$.
  \end{proof}

Lemma~\ref{lem:lower-cells} implies that $\kvd$
contains $\Omega(k(m-k))$ cells, and
Lemma~\ref{lem:lower-edges} implies that $\kvd$
contains $\Omega(\min\{nk, n(m-k)\})$ edges.

  \paragraph{An upper bound.}  As Liu and Lee~\cite{LL-KVD-2013} shows, 
  each Voronoi
  vertex of $\kvd$ has degree 2 or 3, except for
  the vertices lying on the boundary of $P$. In the following, we show that $\kvd$
  has $O(k(m-k)+\min\{nk,n(m-k)\})$ Voronoi edges, which implies that
  the complexity of $\kvd$ is $O(k(m-k)+\min\{nk,n(m-k)\})$ by the Euler's formula for planar graphs.

  A Voronoi edge is a part of the bisecting curve of two
  sites.  Consider Voronoi edges that come from the
  bisecting curve of two sites $s_1$ and $s_2$ in $S$.  Such
  Voronoi edges form a curve that connects two degree-3 (or degree-1) vertices $v$ and
  $u$.  Let $V(s_1,s_2)$ is the set of vertices of $P$ in
  $\pi(s_1',v)\cup\pi(s_1',u)\cup\pi(s_2',v)\cup\pi(s_2',u)$ excluding $s_1'$ and $s_2'$,
  where $s_i'$ is the junction of $\pi(s_i,u)$ and
  $\pi(s_i,v)$ for $i=1,2$.  Recall that the number of
  Voronoi edges on the common boundary of $V_1$ and $V_2$ is
  $O(1+|V(s_1,s_2)|)$. (Refer to Figure~\ref{fig:bisect}(b).)

  For each vertex $w$ of $P$, we focus on the number
  of pairs $(s_1,s_2)$ of sites such that $w$ is
  contained in $V(s_1,s_2)$.  Observe that either $s_1$ or
  $s_2$ (but not both) is one of the $k$ nearest sites
  from $w$. Without loss of generality, we assume that $s_1$ is one of the $k$ nearest sites from $w$.
  We  claim that there is no site $s_3$ distinct from $s_1$ such that
  $w \in V(s_2,s_3)$.  
  This is because $\pi(s_2,x)$ intersects at most one Voronoi edge
  defined by $s_2$ and some other site for any point $x$ in $P$.
  We also claim that there is no site $s_3$ distinct from $s_2$ such that
  $w\in V(s_1,s_3)$. Assume to the contrary that there is such a site $s_3$. 
  We observe that $s_1$ is one of the $k$ nearest sites
  of any point in the region $R$ bounded by $\pi(s_1',v)\cup \pi(s_1',u)$ and $b(s_1,s_2)$. 
  A Voronoi edge defined by $s_1$ and $s_3$ does not intersect $R$.
  This means that $V(s_1,s_2)$ and $V(s_1,s_3)$ are interior disjoint.
  Thus $w$ is the junction of $\pi(s_1,u')$ and $\pi(s_2,v')$, where $u'$ and $v'$ are 
  the endpoints of the Voronoi edge defined by $s_1$ and $s_2$. Thus it is not contained in
  $V(s_1,s_3)$, which is a contradiction.
  By the two claims, the number of pairs $(s,s')$
  with $w \in V(s,s')$ is $\min\{k, m-k\}$ for each vertex $w$ of $P$.

  Therefore, the number of Voronoi edges is
  $O(k(m-k)+\min\{nk, n(m-k)\})$, and the total
  complexity of $\kvd$ is $O(k(m-k)+\min\{nk,
  n(m-k)\})$.

  Combining the lower bound example, we have the
  following lemma.

\begin{lemma}
  The geodesic order-$k$ Voronoi diagram of $m$
  points contained in a simple polygon with $n$ vertices
  has complexity of $\Theta(k(m-k)+\min\{nk,
  n(m-k)\})$.
\end{lemma}

\subsection{Computing the Topological Structure of the Diagram}
We use the notations defined in Section~\ref{sec:NVD}.
For $\vd$, we maintain the combinatorial structure of the beach line as
$x$ moves along $\bd P$. For $\kvd$, we maintain the combinatorial
structures of $k$ curves each of which represents a level of the arrangement
of some curves.

For a site $s \in P(x)$, recall that $R_s(x)=\{p \in P(x): d(p,s)\leq
d(p,\pi(o,x))\}$ is connected.
Moreover, the boundary of $R_s(x)$ consists of one polygonal chain of $\bd P$
and a simple curve with endpoints on $\bd P$.
We call the simple curve the \emph{wave-curve} for $s$. 
A wave-curve is weakly monotone with respect to $\pi(o,x)$ by
Lemma~\ref{lem:R-monotone}. 

We say that a point $p\in P(x)$ lies \emph{above} a curve if $\pi(p,
\pi(o,x))$ intersects the curve.  
We use $L_i(x)$ to denote the region of $P(x)$ consisting of all points
lying above at least $i$ wave-curves for sites contained in $P(x)$. 
The \emph{$i$th-level} of the arrangement of wave-curves of sites contained 
in $P(x)$ is defined to be the boundary of $L_i(x)$ excluding $\bd P$. The
$i$th-level for each $i$ is weakly monotone with respect to $\pi(o,x)$ and
consists of at most $m$ simple curves with endpoints lying on $\bd P$.
 Note that the beach line
for $\vd$ coincides with the 1st-level of the arrangement of the wave-curves for all sites in $P(x)$ and $R(x)$ coincides with $L_1(x)$.

\begin{lemma}
  $\kvd[S]$ restricted to $L_k(x)$ coincides with $\kvd[S\cap P(x)]$
  restricted to $L_k(x)$.
\end{lemma}
\begin{proof}
  Consider a Voronoi cell of $\kvd[S]$ containing a point $p \in L_k(x)$. We
  show that this cell is associated with $k$ sites in $S\cap P(x)$.
  To see this, observe that $d(p,s) > d(p,\pi(o,x))$ for any site $s
  \notin P(x)$.  By the definition of $L_k(x)$, there are $k$ sites in
  $P(x)$ whose geodesic distance from $p$ is at most $d(p,\pi(o,x))$.
  Thus the Voronoi cell of $\kvd[S]$ containing $p$ is associated with
  $k$ sites in $S\cap P(x)$.  Therefore, the lemma holds.
\end{proof}

Therefore, we can obtain the topological structure of $\kvd[S]$ by
maintaining the topological structure of the $k$th-level of the
arrangement of wave-curves.  In addition to this, we maintain the
topological structures of the $i$th-level of the arrangement for all $i<
k$ to detect the changes to the topological structure of the
$k$th-level.  
A breakpoint of the $i$th-level is an intersection of two wave-curves 
for each $i=1,\ldots,k$, and an endpoint of the $i$th-level is an endpoint
of a maximal connected curve in the $i$th-level.
Let $B_i$ be the sequence of breakpoints and endpoints of the $i$th-level, 
which represents the combinatorial structure of the $i$th-level.

To avoid the case that $R(x)$ does not contain $P$ after the sweep, we consider
a simple polygon $P'$ obtained from $P$ by attaching a long and 
very thin triangle on $\bd P$
as we do in Section~\ref{sec:NVD}.

\paragraph{Events.}  There are four types of events: site
  events, circle events, vanishing events and merging events. 
  Every event corresponds to a
\emph{key} which is a point on $\bd P'$.  An event occurs when the
sweep point $x$ passes through its corresponding key.
The definitions of the event types are analogous to the ones in
Section~\ref{sec:NVD}.

Each site $s$ in $S$ defines a \emph{site event}.  The key of the site event
defined by $s$ is the point on $\bd P'$ 
that comes first from $o$ in clockwise order along $\bd P'$ among the points $x'$ 
with $s\in \pi(o,x')$. 
A pair $(\beta_1,\beta_2)$ of consecutive breakpoints in $B_i$ defines a
\emph{circle event} if there exists the point $c$ equidistant from three sites
defining $\beta_1$ or $\beta_2$.  
The key of this circle event is the point on $\bd P'$ that comes first from $o$ in clockwise order along $\bd P'$ 
among the points $x'$ with $d(c,\pi(o,x'))=d(c,p)$,
where $p$ is one of the sites defining $\beta_1$ or $\beta_2$.
Each breakpoint $\beta$ in $B_i$ defines a \emph{vanishing event}. 
Let $(s_1,s_2)$ be the pair defining $\beta$.
Let $p$ be the point on $\bd P'$ equidistant from $s_1$ and $s_2$
that lies outside of $L_i(x)$.
The key $k$ of the vanishing event is the point on $\bd P'$  that comes first from $o$ in clockwise order along $\bd P'$  among the points $x'$ with $d(p,\pi(o,x'))=d(p,s_1)$.
A pair $(\beta_1, \beta_2)$ of 
consecutive endpoints in $B_i$ defines a \emph{merging event} if
$\beta_1$ and $\beta_2$ are endpoints of different connected curves of the $i$th-level.
Let $s_1$ and $s_2$ be the sites defining $\beta_1$ and $\beta_2$, respectively.
Let $p$ be the (unique) point on $\bd P'$
equidistant from $s_1, s_2$ that comes after $\beta_1$ and before $\beta_2$
from $o$ along $\bd P$ in clockwise order.
The key $k$ of the merging event defined by $(\beta_1,\beta_2)$ 
is the point on $\bd P'$ that comes first from $o$ in clockwise order along $\bd P'$  among the points $x'$ 
with $d(p,\pi(o,x'))=d(p,s_1)$.

The changes of $B_i$ due to events are analogous to the ones in Section~\ref{sec:NVD}, except for the circle events.
The sequence $B_i$ may change when the sweep
point passes through the key of some circle event
defined by a pair of consecutive breakpoints in $B_{i-1}$ or $B_{i-2}$
as illustrated in Figure~\ref{fig:higher-events}.

\paragraph{An algorithm.}  Initially, $B_i$ is empty for all 
$i=1,\ldots,k$, so there is no event.  We compute
all site events in advance. For each site, we can compute its key in $O(\log n)$
time~\cite{shortest-path-tree}.  As we handle events, the sequence
$B_i$ changes accordingly and we reflect these changes to $B_i$.  At the same
time, we compute the topological structure of $\kvd$ using $B_k$.

To handle a site event defined by a site $s$, we find the position
for $s$ in $B_i$ and add $s$ to $B_i$ for each $i\leq k$.  This takes
$O(\log^2 n\log |B_i|)$ time for each $i$ by
Lemma~\ref{lem:com-break-site}.  Adding $s$ to $B_i$ makes new pairs
of consecutive breakpoints or endpoints in $B_i$, or 
new breakpoints, which define some events.  We
compute the key for each event 
in $O(\log^2 n)$ time using Lemma~\ref{lem:circle-key}.

\begin{figure}
  \begin{center}
    \includegraphics[width=0.9\textwidth]{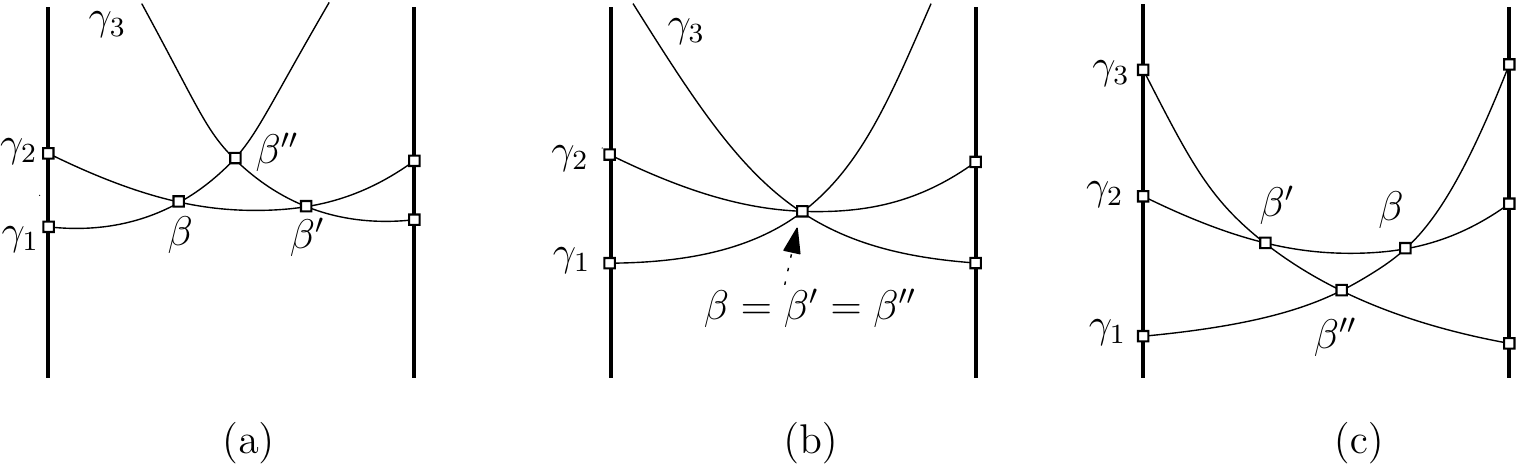}
    \caption{The curves 
      $\gamma_1$, $\gamma_2$, and $\gamma_3$ are wave-curves of $s_1,s_2$, and $s_3$, respectively. 
      The points $\beta, \beta'$, and $\beta''$ are breakpoints
      defined by $(s_1,s_2),(s_2,s_3)$, and $(s_3,s_1)$, respectively.
      \label{fig:higher-events}}
  \end{center}
\end{figure}
To handle a circle event defined by a pair $(\beta,\beta')$ of
consecutive breakpoints in $B_i$, we consider
the following observations, 
which are analogous to the ones for Euclidean order-$k$ Voronoi diagram
given by Zavershynskyi and Papadopoulou~\cite{higherVD3}.
Let $(s_1,s_2)$ and $(s_2,s_3)$ be the pairs of sites defining $\beta$
and $\beta'$, respectively. Let $c$ be the point equidistant from $s_1, s_2$
and $s_3$, and $k$ be the key of the event.  When the sweep
point $x$ passes through $k$, two breakpoints merge into a breakpoint
defined by $(s_1,s_3)$ in $B_i$. 
For illustration, see Figure~\ref{fig:higher-events}.

Additionally, $B_{i+1}$ and $B_{i+2}$ change.
The breakpoints $\beta$ and $\beta'$ are also in $B_{i+1}$ before 
$x$ reaches $k$. Moreover, a breakpoint $\beta''$ defined by $(s_1,s_2)$
lies between $\beta$ and $\beta'$ in $B_{i+1}$.
When $x$ reaches $k$, the three breakpoints merge into one (degenerate) breakpoint equidistant from $s_1,s_2$ and $s_3$.
After $x$ passes through $k$, the order of the breakpoints changes:
$\beta', \beta''$ and $\beta$.
For $B_{i+2}$, there is $\beta''$ before $x$ reaches $k$.
After $x$ passes through $k$, $\beta''$ is replaced with $\beta'$ and $\beta$.
The remaining levels
do not change as $x$ passes through $x'$.  We update $B_i$,
$B_{i+1}$ and $B_{i+2}$ accordingly in $O(\log |B_i|+\log |B_{i+1}|+\log |B_{i+2}|)$ time.

After updating $B_i$, $B_{i+1}$ and $B_{i+2}$, we have new pairs of
consecutive breakpoints or endpoints, or new breakpoints.  
We compute the events defined by them and
add the events to the event queue in $O(\log^2 n + \log |B_i|+\log |B_{i+1}|+\log |B_{i+2}|)$ time.
The other events can be handled in a way similar to the one in Section~\ref{sec:NVD} in $O(\log^2 n + \log |B_i|)$ time.

\paragraph{An analysis.}  We have $m$ site events and spend
$O(\sum_{1\leq i\leq k}\log^2 n\log |B_i|)$ time to handle each site
event, where $|B_i|$ is the maximum length of $B_i$ during the sweep.  We can handle all site events in
$O(km\log^2n \log m)$ time in total by the following lemma.
\begin{lemma}
  The length of $B_i$ is $O(im)$ at any time.
\end{lemma}
\begin{proof}
  For $i=1$, the lemma holds by Lemma~\ref{lem:length}.  Consider
  $i\geq 2$.  Initially, $B_i$ is empty.  When the sweep point
  passes through the key of a site event, the length of $B_i$
  increases by at most two.  When the sweep point passes through
  the key of a circle, vanishing, or merging event caused by $B_i$, the length of
  $B_i$ decreases.  In contrast to the nearest-point Voronoi diagram, there is one more case that
  the length of $B_i$ increases in the higher-order Voronoi diagram.  A circle event due to $B_{i-1}$ or $B_{i-2}$ (or,
  $B_{i-1}$ for $i=2$) increases the length of $B_i$ by at most one.
  Note that a circle event due to $B_j$ that is valid when the sweep
  point passes through its key corresponds to a degree-3
  Voronoi vertex of $j$-$\vd$. Thus there are $O(im)$ circle events 
  due to $B_{i-1}$ or $B_i$  that are valid when the sweep
  point passes through their keys.  Therefore, the length of $B_i$ is
  $O(im)$, and the lemma holds.
\end{proof}

There are $O(k^2 m)$ circle events in total 
that are valid when the sweep point passes through their keys.
This is because each such circle event
corresponds to a degree-3 Voronoi vertex of $i$-$\vd$ for some $i$.  
There are $O(k^2 m)$ vanishing or merging events in total that are
valid when the sweep point passes through their keys.
This is because each such event corresponds to a degree-1 Voronoi vertex. 
Recall that $i$-$\vd$ has $O(i(n-i))$ degree-1 or
degree-3 vertices. 
Since we can handle each event in $O(\log m\log^2 n)$ time,
the running time for this algorithm is $O(k^2m\log m\log^2 n)$ time.

Therefore, we have the following lemma and theorem.
\begin{lemma}
  After constructing the shortest path data structure for $P'$ and the shortest path
  map for $P'$ from any fixed point, we can compute $\kvd$ in 
  $O(k^2m\log m\log^2 n)$ time using $O(km)$ space (excluding the two data
  structures).
\end{lemma}

\begin{theorem}
  Given a set of $m$ points contained in a simple polygon with $n$
  vertices, we can compute the geodesic order-$k$ Voronoi of the
  points in $O(k^2m\log m\log^2 n +\min\{nk, n(m-k)\})$ time using
  $O(n+km)$ space.
\end{theorem}

\section{The Geodesic Farthest-Point Voronoi Diagram}
\label{sec:FVD}
In this section, we present two algorithms for computing the
topological structure of $\fvd$: an $O(n+m\log m+m\log^2n)$-time
algorithm and an $O(m\log m\log^2 n)$-time algorithm. The latter
algorithm assumes that we have the shortest path data structure for $P$.  The
former algorithm is faster than the latter because computing the
shortest path data structure for $P$ takes $O(n)$ time.  However, the latter
algorithm has an advantage if we compute $\fvd$ several times
with different point sets.  For an application, see
Section~\ref{sec:application}.

\paragraph{The algorithm given by Aronov et al.}  Aronov et
al.~\cite{FVD} showed that $\fvd$ forms a tree whose root is
the geodesic center $c$ of the sites.  They first compute $c$ and
the geodesic convex hull of the sites.  Then they compute $\fvd$
restricted to the boundary of the polygon in $O((n+m)\log(n+m))$ time.
That is, they compute all degree-1 vertices of $\fvd$.

Then they compute the edges of $\fvd$ towards the geodesic center by a
\emph{reverse geodesic sweep method}.  They showed that $d(u,c) \leq
d(v,c)$ for any ancestor $u$ of a Voronoi vertex $v$ in $\fvd$ (a
tree).  Thus they compute the Voronoi edges of $\fvd$ from the
boundary of $P$ toward $c$ one by one.  This takes $O((n+m)\log(n+m))$
time.

\paragraph{Computing the geodesic center of the sites and the
  geodesic convex hull of the sites.}  We follow the framework of the
algorithm by Aronov et al.~\cite{FVD}, but we can achieve a faster
algorithm for $m \leq n/\log^2 n$ by applying their approach to compute the
topological structure of $\fvd$.  We first compute the geodesic center
of the sites in $S$. This takes $O(n+m\log m)$
time~\cite{1-center-jnl,shortest-path}, or $O(m\log m\log^2 n)$ time with
the shortest path data structure for $P$ by Theorem~\ref{thm:center-point}.
The geodesic convex hull of the sites can be computed in
$O(n+ m\log m)$ time~\cite{shortest-path}.

\paragraph{Computing the diagram restricted to the boundary of the
  polygon} We can compute $\fvd$ restricted to the boundary of $P$ in
$O(n+m)$ time using the algorithm by Oh et al.~\cite{fvd_boundary}
once we have the sequence of the
sites along the boundary of the geodesic convex hull.  The
algorithm uses the property that the order of the
sites along the boundary of the geodesic convex hull of them is the
same as the order for their Voronoi cells along the boundary of the polygon,
which is shown by Aronov et al~\cite{FVD}.  
They consider the sites on the boundary of the geodesic convex hull
one by one and compute the $\fvd$ restricted to the boundary of $P$
along the boundary.
Combining their approach with Lemma~\ref{lem:equidistant-boundary}, we
can compute $\fvd$ restricted to the boundary of $P$ in $O(m\log^2 n)$
time.

\paragraph{Extending the diagram inside the polygon} We apply the
reverse geodesic sweep from the boundary of $P$ towards the center $c$
of the sites.  In this problem, a \emph{sweep line} is a simple curve
consisting of points equidistant from $c$.  Let $B$ be the (circular)
sequence of the sites that have their Voronoi cells on the sweep line
in clockwise order.  As an exception, in the initial state, we set $B$
to be the sequence of the sites whose Voronoi cells appear on $\bd P$.

No site, vanishing, or merging event occurs during the sweep
because $\fvd$ forms a tree. 
So we handle the circle event
only.  Every triple of consecutive sites in $B$ defines a circle
event.  The key defined by a triple of sites is $d(c,c')$ for
the point $c'$ equidistant from the three sites.  Thus, we can
compute the key of each triple of consecutive sites in $B$ in
$O(\log^2 n)$ time.

The analysis is similar to the one in Section~\ref{sec:NVD}.
There are $O(m)$ events in total, and each event can be handled in $O(\log m+\log^2 n)$ time.
Thus extending the diagram inside the polygon takes $O(m\log m+m\log^2 n)$ time.

\begin{lemma}
  We can compute the topological structure of $\fvd$ in 
  $O(m\log m\log^2 n )$ or $O(n+m\log m+m\log^2 n)$ time once the
  shortest path data structure for $P$ is constructed.
\end{lemma}

\begin{theorem}
  Given a set of $m$ point sites contained in a simple polygon with $n$
  vertices, we can compute the geodesic farthest-point Voronoi diagram
  of the points in $O(n+ m\log m+m\log^2 n)$ time using $O(n+m)$
  space.
\end{theorem}

\section{Dynamic Data Structures for Nearest or Farthest Point Queries}
We showed that the topological structure of
a Voronoi diagram can be computed without considering the whole
polygon once we have the shortest path data structure for $P$.  
The topological structure can also be used for data structures for answering
nearest or farthest point queries for a dynamic point set.

We apply the framework given by Bentley and Saxe~\cite{decomposable}.
At all times, we maintain $O(\sqrt m)$ disjoint sets each of which
consists of $O(\sqrt m)$ points.  For each set, we compute the topological
structure of the nearest-point (or the farthest-point) Voronoi diagram of the points
in the set. This takes $O(\sqrt{m}\log m\log^2n)$ time.
For a nearest (or a farthest) point query, we simply find the 
Voronoi cell containing the
query point for each of $O(\sqrt m)$ Voronoi diagrams.
Then we have $O(\sqrt m)$ candidates for the nearest (or the farthest) point
from the query point. We find the nearest (or the farthest) point from
the query point directly among them.

While processing updates, we maintain the invariant that every
set, except only one, contains at most $2\sqrt{m}$ and at least $\sqrt{m}/2$ points at any time,
where $m$ is the number of the points in the moment.
For the insertion of a
point $p$, we choose the set of smallest size.
If the set contains more than $2\sqrt{m}$ points, we split it into two equal sized sets,
add $p$ to one of them, and
reconstruct the Voronoi diagram for each set. Otherwise, we add $p$ to the set of smallest size
and reconstruct the Voronoi diagram for this set. The invariant holds for any case.  
 For the deletion of a
point $p$, we find the set containing $p$ and remove $p$ from the set.
If there are two sets containing less than $\sqrt{m}/2$ points, we merge them into one 
and reconstruct the Voronoi diagram for the set. The invariant still holds.
By the invariant, we have $O(\sqrt{m})$ sets at any time. For each update, we reconstruct
the Voronoi diagrams of at most two sets containing at most $2\sqrt{m}$ points. Thus 
the update time is $O(\sqrt{m}\log m\log^2 m)$.


To answer a query, we find the Voronoi cell containing the query point
using the topological structure of a Voronoi diagram.
To do this, when computing the adjacency graph, we also construct a data  structure
that supports a point location query as follows.

\paragraph{Point location.}
The key idea is to 
approximate the Voronoi diagram
into a polygonal subdivision using its adjacency graph and the exact
positions of degree-1 and degree-3 vertices.
Consider an edge of the adjacency graph corresponding 
to two adjacent Voronoi cells $V_1$ and $V_2$.
We say that a polygonal curve \emph{approximates} the common boundary of 
$V_1$ and $V_2$ if it connects the two endpoints of the common boundary, 
consists of at most three
line segments, and is contained in the closure of $V_1\cup V_2$.

\begin{figure}
  \begin{center}
    \includegraphics[width=0.9\textwidth]{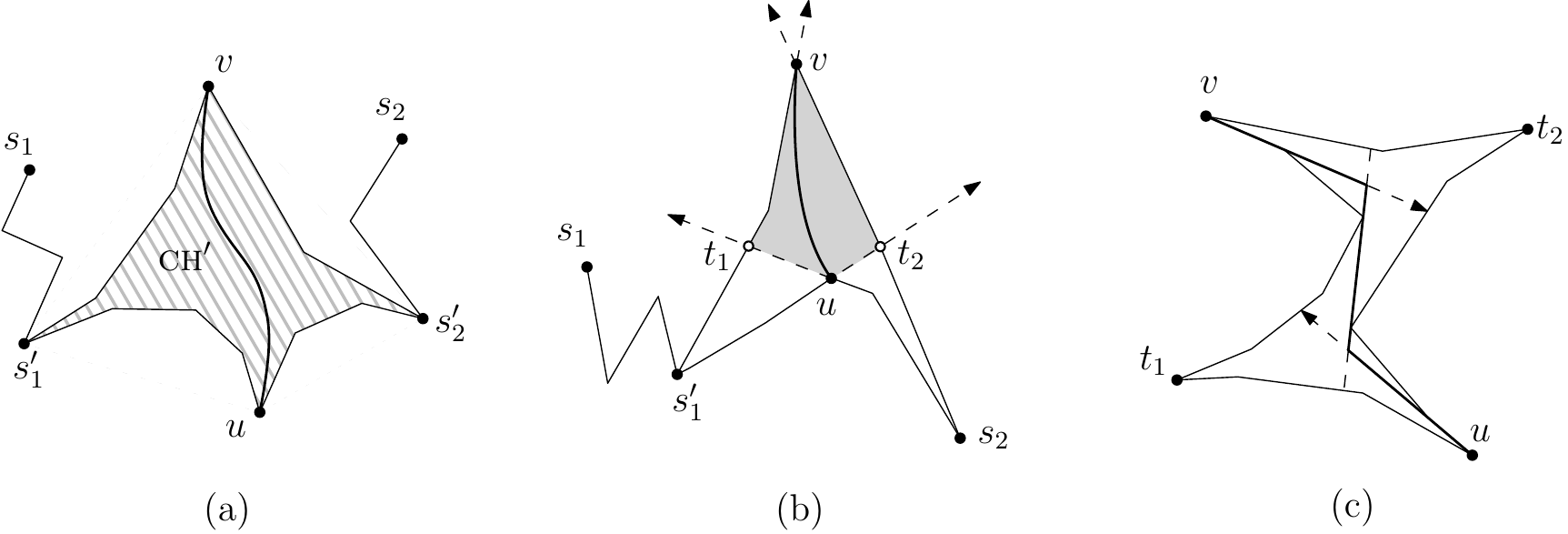}
    \caption{\label{fig:point-location}
      (a) The common boundary of the Voronoi cells with sites $s_1$ and $s_2$
        is contained in the convex hull $\ch'$ (dashed region)
      of $s_1',s_2',u$ and $v$.
      (b) For the farthest-point Voronoi diagram,
      we extend four edges to find a region
      that is contained in the closure of $V_1\cup V_2$ and is bounded by four concave
      curves.
      (c) Given a region bounded by four concave curves,
      we can find a polygonal curve approximating the common boundary.}
  \end{center}
\end{figure}

\begin{lemma}
  \label{lem:approx-edge}
  We can compute a polygonal curve approximating 
  the common boundary of two adjacent Voronoi cells in $O(\log n)$ time.
\end{lemma}
\begin{proof}
  Let $u$ and $v$ be the endpoints of the common boundary of
  two adjacent Voronoi cells $V_1$ and $V_2$.
  Let $s_i$ be the site associated with $V_i$ and $s_i'$
  be the junction of $\pi(s_i,u)$ and $\pi(s_i,v)$ for $i=1,2$.
  The junction of $\pi(u,s_1)$ and $\pi(u,s_2)$ is $u$ itself, and 
  the junction of $\pi(v,s_1)$ and $\pi(v,s_2)$ is $v$ itself
  under the general position condition. See figure~\ref{fig:point-location}(a) and (b).
  Therefore, the geodesic convex hull, denoted by $\ch'$, of $s_1',s_2',u$ and $v$
  consists of at most four maximal concave curves.
  
  Using this fact, we first find two points $t_1, t_2$ such that
  the geodesic convex hull $\ch$ of  $t_1,t_2, u, v$ is
  contained in the closure of $V_1\cup V_2$ and the boundary of $\ch$
  consists of at most four maximal concave polygonal curves.
  Then we compute a polygonal curve approximating the common boundary
  in $\ch$.

  We first show how to find such points $t_1$ and $t_2$.
  Consider the case that $V_1$ and $V_2$ are two adjacent Voronoi cells
  in the nearest-point Voronoi diagram.
  Observe that $\pi(s_i',u)\cup\pi(s_i',v) \subseteq V_i$ for $i=1,2$.
  Thus, $\ch'$ is contained in the closure of $V_1\cup V_2$.
  Since the boundary of $\ch'$ consists of at most four concave polygonal curves,
  $\pi(s_i',u), \pi(s_i',v)$ for $i=1,2$, the two points 
  $s_1'$ and $s_2'$ satisfy the condition for $t_1$ and $t_2$, respectively.
	
  For the case that $V_1$ and $V_2$ are two adjacent Voronoi cells in the
  farthest-point Voronoi diagram, 
  consider two line segments contained in $P$ with endpoint $u$ such that
  each of them is collinear to the edge of $\pi(s_i,u)$ incident to $u$ for $i=1,2$.
  Similarly, we consider two line segments contained in $P$ with endpoint $v$ such that
  each of them is collinear to the edge of $\pi(s_i,v)$ incident to $v$ for $i=1,2$.
  If $u$ or $v$ is on $\bd P$, the line segments are $u$ or $v$ itself.
  It is known that the line segments incident to $u$ are contained in $V_1$ and 
  the other two line segments are contained in $V_2$~\cite{FVD}.
    The four line segments subdivide $P$ into subpolygons.
    Let $Q$ denote the subpolygon
    that contains the common boundary of $V_1$ and $V_2$. If $u$ and $v$ are on $\bd P$,
    the subpolygon $Q$ is the whole polygon $P$. 
    Then $Q$ is contained in the closure of $V_1\cup V_2$.
See Figure~\ref{fig:point-location}(b).

Consider the intersection $Q\cap \ch'$.
  For $i=1,2$,
  if $s_i'$ is contained in $Q$, both $\pi(s_i',u)$ and $\pi(s_i',v)$
  appear on the boundary of $Q\cap\ch'$, and we set $t_i=s_i'$.
  Otherwise, only one of them, say $\pi(s_i',v)$, crosses one of the line segment 
  appearing on the boundary of $Q$. We set the cross point to $t_i$. See Figure~\ref{fig:point-location}(b).
  Therefore, 
  $Q\cap \ch'$ is the geodesic convex hull of $t_1, t_2, u$ and $v$ whose boundary
  consists of at most four maximal concave polygonal curves,
  $\pi(t_i,u)$ and $\pi(t_i,v)$ for $i=1,2$.
  Moreover, we can compute $t_i$ 
  in $O(\log n)$ time using Lemma~\ref{lem:extension}.
  
  Now, we have two points $t_1, t_2$ such that the 
  geodesic convex hull $\ch$ of $t_1, t_2, u, v$ is contained in
    the closure of $V_1\cup V_2$ and consists of at most four maximal concave polygonal
  curves. Since $\pi(u,v)\subset\ch$, 
  we approximate $\pi(u,v)$ by
  a polygonal curve consisting of at most three segments contained in $\ch$ as follows.
  If $\pi(u,v)$ is a line segment, it is the polygonal curve. Otherwise, 
  we can compute a polygonal curve consisting of at most three line segments
  from the extensions of the two edges of $\pi(u,v)$, one incident to $u$ and
  one incident to $v$, towards
  the other endpoints of $\pi(u,v)$, and one line segment tangent to $\pi(t_1,v)$ and
  $\pi(t_2,u)$.
  See Figure~\ref{fig:point-location}(c).
  This can be done in
  $O(\log n)$ time using the shortest path data structure for $P$.
  This proves the lemma for both the nearest-point and farthest-point 
  Voronoi diagrams.
\end{proof}

Since each polygonal curve approximating the common boundary of two 
adjacent Voronoi cells described in Lemma~\ref{lem:approx-edge} is contained in
the closure of the union of the Voronoi cells,
the following property is satisfied.
\begin{corollary}
  No two polygonal curves constructed by Lemma~\ref{lem:approx-edge}
  cross each other.
\end{corollary}

The following lemma allows us to construct a data structure on a
adjacency graph that supports a point location query efficiently.
In the lemma, we assume that we have a ray-shooting data structure for $P$.
\begin{lemma}
  \label{lem:point-location}
  We can construct a data structure for the topological structure of 
  the nearest-point (or the farthest-point) 
  Voronoi diagram of $O(\sqrt{m})$ points in $O(\sqrt{m}\log (n+m))$ time
  that supports a point location query for the Voronoi
  diagram in $O(\log (n+m))$ time.
\end{lemma}
\begin{proof}
  We first approximate every common boundary of two adjacent Voronoi cells
  by a polygonal curve, which takes $O(\sqrt{m}\log n)$ time in total.
  Then we have $O(\sqrt{m})$ polygonal curves each of which
  consists of at most three line segments.
  Let $A$ be the set of all line segments of the curves.
  As a preprocessing, we construct a ray-shooting 
  data structure for $A$ in $O(\sqrt m\log m)$ time
  using $O(\sqrt m)$ space. 
  
  We show how to find the Voronoi cell containing a query point $q$.
  Consider the subdivision $\mathcal{P}_A$ of $P$ with respect to $A$.
  Note that the complexity of $\mathcal{P}_A$ is $O(\sqrt m + n)$.
  We first find the region in $\mathcal{P}_A$ containing $q$ as follows
  without constructing $\mathcal{P}_A$ explicitly.
  Let $r$ be the ray from $q$ going upwards.
  We find the line segment in $A$ that $r$ hits first
  using the ray-shooting data structure constructed on $A$ in $O(\log m)$ time.
  Then we find the edge of $P$ that $r$ hits first in $O(\log n)$ time.
  By comparing the vertical distances from $q$ to the edge
  of $P$ and the line segment in $A$, we can determine the
  edge of $\mathcal{P}_A$ that $r$ hits first.
  Thus we can find the region in $\mathcal{P}_A$ containing $q$ in $O(\log (n+m))$ time. (If $r$ hits an edge of $P$ before hitting any approximate polygonal curve,
  we determine which region in $\mathcal{P}_A$ contains $q$ in $O(\log (n+m))$ time
  by sorting the degree-1 vertices along $\bd P$ in advance.)
  
  Let $s$ be the site associated with the region we just found.
  Since we consider only the subdivision $\mathcal{P}_A$, $s$ might not be
  the nearest-point (or the farthest-point) from $q$.
  In this case, $s$ is contained in a region bounded by a chain $\gamma$
  of Voronoi edges defined by $s$ and $s'$,  
  and a chain $\tilde{\gamma}$ approximating $\gamma$. Note that 
  $s'$ is the nearest-point (or the farthest-point) of $q$.
  This is because a polygonal curve approximating the common boundary of two Voronoi cells
  is contained in the union of them, and no 
  two curves constructed from Lemma~\ref{lem:approx-edge} cross 
  each other.
  
  We show how to find $s'$.
  For a nearest-point query, the ray from $q$ in direction opposite
  to the edge of $\pi(s,q)$ incident to $q$
  does not intersect $\gamma$.
  Since $q$ is contained in the region bounded by $\gamma$ and $\tilde{\gamma}$,
  but the ray does not intersect $\gamma$,
  the ray intersects $\tilde{\gamma}$.
  We compute $\tilde{\gamma}$ by finding the point where 
  the ray hits first in $\mathcal{P}_A$ as we did
  before, and find the site $s'$ defining $\tilde{\gamma}$ other than $s$.
  By comparing $d(s,q)$ and $d(s',q)$, we get the answer in $O(\log (n+m))$ time.
  
  For a farthest-point query, we do in a similar way.
  In this case, the extension of the edge of $\pi(s,q)$ incident to $q$ in direction opposite
  to $q$ does not intersect $\gamma$, and thus it
  intersects $\tilde{\gamma}$.
  We find the edge of $\mathcal{P}_A$ that
  intersects the extension in $O(\log (n+m))$ time,
  and compute $s'$ in a way similar to the one for
  a nearest-point query.
  Then by comparing $d(s,q)$ and $d(s',q)$, we get the answer in $O(\log (n+m))$ time.
\end{proof}

We perform the point location query on each of the $O(\sqrt m)$ Voronoi diagrams
as described in Lemma~\ref{lem:point-location} and find
the nearest (or the farthest) point from the query point among the $O(\sqrt m)$
solutions, which takes $O(\sqrt m\log (n+m))$ time. For the update,
we need to compute  the topological structure of the Voronoi diagram for
the updated set, which takes $O(\sqrt{m}\log m\log^2n)$ time.
Then we reconstruct the data structure for point location on the topological structure.

\begin{theorem}
  We can construct a data structure of size $O(n+m)$ that supports a
  nearest-point (or a farthest-point) query under point insertions and
  deletions.  Each query takes
  $O(\sqrt m\log (n+m))$ time and each update takes $O(\sqrt m\log m\log^2
  n)$ time.
\end{theorem}

\label{sec:application}
\bibliography{paper}{}
\end{document}